\documentclass[journal]{IEEEtran}
\usepackage[utf8]{inputenc}
\usepackage{amsmath}
\usepackage{soul}
\usepackage{array}
\usepackage{amsmath,amsthm,amssymb,epsf}
\usepackage{bm} 
\usepackage{graphicx,psfrag}
\usepackage{epstopdf}
\usepackage{verbatim}
\usepackage{multirow}
\usepackage{color}
\usepackage{setspace}
\usepackage{enumitem}
\usepackage{subfigure}
\usepackage{amssymb}
\usepackage{tabularx}
\usepackage{booktabs}
\usepackage{dirtytalk}
\usepackage{amsmath}
{
      \theoremstyle{plain}
      \newtheorem{assumption}{Assumption}
  }
\usepackage{csquotes}
\usepackage{algorithm}
\usepackage{algorithmic}

\newtheorem{theorem}{Theorem}
\newtheorem{definition}{Definition}
\newtheorem{problem}{Problem}
\newtheorem{lemma}{Lemma}

\newtheorem{remark}{Remark}
\newtheorem{condition}{Condition}
\newtheorem{corollary}{Corollary}
\usepackage{soul}
\usepackage{cite}

\begin{document}

\title{Safe Data-Driven Predictive Control}

\author{Amin~Vahidi-Moghaddam, Kaian~Chen, Kaixiang~Zhang, Zhaojian~Li$^*$, Yan~Wang, and Kai~Wu
\thanks{This work was supported by the Ford Motor Company under award number 000315-MSU0173.}
\thanks{$^*$Zhaojian Li is the corresponding author.}
\thanks{Amin Vahidi-Moghaddam, Kaian Chen, Kaixiang Zhang, and Zhaojian Li are with the Department of Mechanical Engineering, Michigan State University, East Lansing, MI 48824 USA (e-mail: vahidimo@msu.edu, chenkaia@msu.edu, zhangk64@msu.edu, lizhaoj1@egr.msu.edu).}
\thanks{Yan Wang and Kai Wu are with the Research and Advanced Engineering, Ford Motor Company, Dearborn, MI 48121 USA (e-mail: ywang21@ford.com, kwu41@ford.com).}
} 

\maketitle

\begin{abstract}
In the realm of control systems, model predictive control (MPC) has exhibited remarkable potential; however, its reliance on accurate models and substantial computational resources has hindered its broader application, especially within real-time nonlinear systems. This study presents an innovative control framework to enhance the practical viability of the MPC. The developed safe data-driven predictive control aims to eliminate the requirement for precise models and alleviate computational burdens in the nonlinear MPC (NMPC). This is achieved by learning both the system dynamics and the control policy, enabling efficient data-driven predictive control while ensuring system safety. The methodology involves a spatial temporal filter (STF)-based concurrent learning for system identification, a robust control barrier function (RCBF) to ensure the system safety amid model uncertainties, and a RCBF-based NMPC policy approximation. An online policy correction mechanism is also introduced to counteract performance degradation caused by the existing model uncertainties. Demonstrated through simulations on two applications, the proposed approach offers comparable performance to existing benchmarks with significantly reduced computational costs.
\end{abstract}

\textbf{Note to Practitioners}. Model predictive control (MPC) implementation causes high computational cost due to solving an online optimization problem at each time step. To prevail over this challenge, one prevalent approach involves the utilization of model-reduction techniques to simplify the system dynamics \cite{zamani2024data}. However, even after the model reduction, the computational burden often remains substantial. Another sound approach involves the implementation of function approximators, such as the NNs \cite{krishnamoorthy2021adaptive} and the GPR \cite{arcari2020meta}, to learn the MPC policy. However, the trained controllers suffer from several shortcomings: i) Unlike the MPC that ensures system safety, trained controllers offer no assurance on safety, and ii) The trained controllers inevitably cause performance degradation due to the model uncertainties. These issues have significantly hindered the adoption of such controllers in real-world engineering applications.

\begin{IEEEkeywords}
 Nonlinear Model Predictive Control; Robust Control Barrier Function; Spatial Temporal Filters; Concurrent Learning; Nonlinear System Identification.
\end{IEEEkeywords}

\IEEEpeerreviewmaketitle

\section{Introduction}
\label{Sec1}
Model predictive control (MPC) provides an optimal control policy with system safety guarantees \cite{seel2023variance,rosolia2023model,zhao2023efficient}. However, the MPC relies on an accurate dynamical model, which is challenging to obtain, especially for nonlinear and complex systems \cite{zhong2023robust,cao2020neural,lopez2019dynamic}. Towards that end, in our prior work \cite{chen2020online}, we have developed a nonlinear system identification that employs spatial-temporal filters (STFs) to represent a nonlinear system as a composite local model structure. Compared to black-box models, based on neural network (NN) \cite{mehrabi2024adaptive,talebzadeh2024deep} and Gaussian process regression (GPR) \cite{coleman2023exploration,ashtiani2024reconstructing}, the STF-based composite local model structure has simpler form with greater interpretability, and the resultant MPC shows a reasonable performance \cite{chen2022stochastic}. Despite promising empirical results, the STF needs to satisfy the persistence of excitation (PE) condition, which plays an inevitable role for the performance of the trained models and will be treated in this paper.

Moreover, control barrier function (CBF) has been introduced to satisfy the system safety constraints such that forward invariance of a safety set is guaranteed using the system dynamics \cite{ames2019control,chen2020safety,dean2020guaranteeing}. Adaptive CBF (ACBF) and robust CBF (RCBF) frameworks have been introduced against system
uncertainties, i.e., state estimation error, model mismatch, and unknown disturbance \cite{taylor2020adaptive,lopez2020robust,garg2021robust,black2021fixed,isaly2021adaptive}. On the other hand, to minimize the performance loss caused by the control learning
error, an online adaptive control policy has been proposed for the MPC policy learning \cite{krishnamoorthy2021adaptive}. However, a holistic treatment of the model uncertainties – necessary for the real-world engineering systems – has not been developed.

In this work, we introduce a safe data-driven predictive control to efficiently track a desired reference trajectory for general nonlinear systems. Specifically, we propose a discrete-time STF-based concurrent learning to identify the system dynamics. Compared to noise-injection PE satisfaction schemes, the concurrent learning uses a memory of past data to satisfy a rank condition which is easy to check online \cite{chowdhary2010concurrent}. An extended RCBF scheme is further developed to guarantee the system safety by systematically considering all model uncertainties. The nonlinear MPC (NMPC) is employed to incorporate the composite local model structure and the RCBF constraint, and the STF function approximator is used one more time to learn the NMPC policy. Finally, a policy correction scheme is proposed for efficient online implementation.

\begin{table}[!ht]
\centering
 \caption{Abbreviations}
\begin{tabular}{ |p{2cm}|p{6cm}| }
\hline\hline
Abbreviations & Meaning \\
\hline
MPC & Model Predictive Control \\\hline
NMPC & Nonlinear Model Predictive Control \\\hline
STF & Spatial-Temporal Filter \\\hline
NN & Neural Network \\\hline
GPR & Gaussian Process Regression \\\hline
PE & Persistence of Excitation \\\hline
CBF & Control Barrier Function \\\hline
RCBF & Robust Control Barrier Function \\\hline
ACBF & Adaptive Control Barrier Function \\\hline
CL & Concurrent Learning \\\hline
RLS & Recursive Least Squares \\\hline
UUB & Uniformly Ultimately Bounded \\\hline
QP & Quadratic Programming \\\hline
\hline
\end{tabular}
\end{table}

The contributions of this paper are pointed out as follows. First, the proposed STF-based concurrent learning handles both structured and unstructured uncertainties as well as unknown external disturbances without requiring derivatives of system states or filter regressors to remove the PE condition compared to \cite{chowdhary2010concurrent}. Second, the developed RCBF guarantees the system safety in the presence of not only the system identification error and the external disturbance but also the control learning error compared to \cite{taylor2020adaptive,lopez2020robust,garg2021robust,black2021fixed,isaly2021adaptive}. Third, an online adaptive control policy, including a KKT adaptation and a feedback control, is proposed to treat the performance loss because of the model uncertainties such that it keeps the real trajectory around the ideal nominal trajectory. Last but not least, the efficacy of the developed control synthesis is illustrated in applications of cart-inverted pendulum and automotive powertrain control. Table I represents the used abbreviations in this paper.



\section{Preliminaries and Problem Formulation}
\label{Sec2}
In this section, we first review the foundational concepts of model predictive control (MPC) and control barrier function (CBF). It subsequently outlines the challenge of safe data-driven predictive control, specifically targeting nonlinear discrete-time systems.

\subsection{Model Predictive Control}
Consider a class of nonlinear discrete-time system that has the following form: 
\begin{equation}
    \label{system}
    x(k+1) = f(x(k),u(k))+w(k),\quad y(k) = g(x(k)),\\
\end{equation}
where $k \in \mathbb{N}^+$ is the time step, $x \in \mathbb{R}^n$ denotes the state vector, $u \in {\mathbb{R}^m}$ represents the control input, $w \in \mathbb{R}^n$ is an unknown external disturbance, and $y \in \mathbb{R}^l$ denotes the output of the system. Moreover, $f:\mathbb{R}^n\times \mathbb{R}^m \rightarrow \mathbb{R}^n$ is the system dynamics, and $g:\mathbb{R}^n \rightarrow \mathbb{R}^l$ represents the output dynamics. Note that the system states $x$ are not measurable.

Now, consider the state $x$ and the control input $u$ under the following constraints:
\begin{equation}
    \label{Constraints}
    u(k) \in \mathbb{U} \subset \mathbb{R}^m,
     \quad x(k) \in \mathbb{X} \subset \mathbb{R}^n,
\end{equation}


\begin{definition}[Closed-Loop Performance]
\label{def1}
Consider the nonlinear system \eqref{system} and a control problem of tracking a reference trajectory $r$ by the output $y$. Starting from an initial state $x_0$, the closed-loop system performance over $N$ steps is characterized by the following cost term:
\begin{equation}
    \label{Cost}
    \begin{aligned}
  &J_N(\bf{x},\bf{u}) \\&= \sum^{N-1}_{k=0} \phi(x(k),u(k),y(k),r(k)) + \psi(x(N),y(N),r(N)),
  \end{aligned}
\end{equation}
where $\mathbf{u} = \left[ u(0),\, u(1),\, \cdots,\, u(N-1)  \right]$, $\mathbf{x} = \left[ x(0),\, x(1),\, \cdots,\, x(N)  \right]$, and $\phi(x,u,y,r)$ and $\psi(x,y,r)$ respectively represent the stage and terminal costs that take the following forms:
\begin{equation}
  \begin{aligned}
    \label{Cost format}
    & \phi(x,u,y,r) = x^{T} Q x + u^{T} R u + (y-r)^{T} P (y-r), \\
    & \psi(x,y,r) = x^{T} Q_N x + (y-r)^{T} P_N (y-r),
  \end{aligned}
\end{equation}
where $Q$, $R$, $P$, $Q_N$, and $P_N$ are positive-definite matrices of appropriate dimensions.
\end{definition}

The MPC aims at optimizing the closed-loop performance \eqref{Cost} while adhering to the constraints \eqref{system} and \eqref{Constraints}. However, in practice, the real nonlinear system \eqref{system} may not be available; thus, system identification algorithms are typically used to achieve an identified (nominal) model as:
\begin{equation}
    \label{nominal model}
     \hat{x}(k+1) = \hat{f}(\hat{x}(k), u(k)),\quad
     \hat{y}(k) = g(\hat{x}(k)),
\end{equation}
where $\hat{x}$, $\hat{y}$, and $\hat{f}$ denote the  states,  outputs, and dynamics of the identified model, respectively.

Therefore, at each time step $k$, the MPC uses a constrained optimization problem as follows:
\begin{equation}
  \begin{aligned}
    \label{NMPC}
    &\qquad\qquad(\mathbf{x}^{*},\mathbf{u}^{*}) = \underset{\mathbf{\hat{x}},\mathbf{u}}{\arg\min} \hspace{1 mm} J_N(\mathbf{\hat{x}},\mathbf{u})\\
    &s.t.\quad \hat{x}(k+1) = \hat{f}(\hat{x}(k),u(k)),\quad \hat{y}(k) = g(\hat{x}(k))\\
    & \hspace{8.5 mm} \hat{x}(0) = x(k),\quad u(k) \in U, \quad \hat{x}(k) \in X,
  \end{aligned}
\end{equation}
In the standard MPC, the first optimal control $u^*(0)$ is put into action, making the system to progress by one step, after which the sequence restarts. Despite its wide-ranging accomplishments in various domains, the MPC does encounter certain persistent challenges. These challenges encompass substantial computational overhead, particularly concerning nonlinear systems, and the need to ensure robustness against the model uncertainties. These issues will be elaborated upon in the forthcoming Section~III.

\subsection{Control Barrier Function}
Contemplate a closed set denoted as $S$, which finds its definition as sublevel set of a continuously differentiable function $h: \mathbb{X} \subset \mathbb{R}^n \rightarrow \mathbb{R}$:
\begin{equation}
\begin{aligned}
    \label{Safe set}
  & S = \{x \in \mathbb{R}^n: h(x) \leq 0\},\\
  & \partial S = \{x \in \mathbb{R}^n: h(x) = 0\},\\
  & Int(S) = \{x \in \mathbb{R}^n: h(x) < 0\}.
  \end{aligned}
\end{equation}
where $S$ is a safe set for the system, i.e., is forward invariant, if $h(x)$ is a CBF as:
\begin{equation}
\begin{aligned}
    \label{CBF condition}
  \triangle h(x(k)) \leq -\gamma h(x(k)),\quad 0 < \gamma \leq 1,
  \end{aligned}
\end{equation}
where $\triangle h(x(k)):=h(x(k+1))-h(x(k))$, and $\gamma$ denotes a design parameter. Based on \eqref{CBF condition}, the upper bound of the CBF exhibits exponential decay at a rate of $1-\gamma$ as:
\begin{equation}
\begin{aligned}
    \label{CBF condition2}
  h(x(k+1)) \leq (1-\gamma) h(x(k)).
  \end{aligned}
\end{equation}

The CBF is widely used in engineering systems to avoid unsafe regions and guarantee the system safety; however, robustness to the model uncertainties must be addressed, which will be treated in Section III-B.

\subsection{Problem Statement}
The nonlinear MPC (NMPC) \eqref{NMPC} poses two major challenges. First, solving a nonlinear optimization problem at each time step incurs heavy computations for the real system. Second, the control performance highly relies on the accuracy of the identified model \eqref{nominal model}. Thus, we aim at developing a safe data-driven predictive control framework to address the mentioned challenges. Specifically, to overcome the first issue, we learn the NMPC policy using a function approximator such that we have $\tilde{u}\approx\pi_{\text{MPC}}$, where $\tilde{u}$ and $\pi_{\text{MPC}}$ represent the NMPC policy approximation and the NMPC policy, respectively. To address the second issue, we propose an online policy correction to robustly handle the model uncertainties. More specifically, from \eqref{system}, \eqref{nominal model}, and $\tilde{u}$, the following drift dynamics errors are defined:
\begin{equation}
  \begin{aligned}
    \label{identification error}
    &e_{s}(k) = \hat{f}(\hat{x}(k),u(k)) - f(x(k),u(k)),\\ &e_{c}(k) = \hat{f}(\hat{x}(k),\tilde{u}(k)) - \hat{f}(\hat{x}(k),u(k)),
    \end{aligned}
\end{equation}
where $e_{c}$ and $e_{s}$ represent the drift dynamics errors due to the imperfections on the NMPC policy learning for the identified model \eqref{nominal model} and the system identification for the nonlinear system \eqref{system}, respectively. As a result, one can express the real system \eqref{system} as:
\begin{equation}
  \begin{aligned}
    \label{disturbed system}
    x(k+1) = \hat{f}(\hat{x}(k),\tilde{u}(k)) - e_{c}(k) - e_{s}(k) + w(k).
  \end{aligned}
\end{equation}

Note that due to the system identification error $e_{s}$ and the unknown disturbance $w$, the NMPC \eqref{NMPC} may not optimize the closed-loop performance \eqref{Cost} for the real system \eqref{system}. Moreover, the control policy learning error $e_{c}$ further exacerbates the complexity to achieve a satisfactory performance for the real system. Therefore, the approximated control policy $\tilde{u}$ is adapted online to mitigate the performance loss caused by $e_{c}$, $e_{s}$, and $w$. The objective is stated as follows.

\begin{problem} [Safe Data-Driven Predictive Control]
\label{prob1}
Consider the nonlinear system \eqref{disturbed system} with the safety constraint \eqref{CBF condition2}. Design a safe data-driven predictive control to achieve the following properties:

\hspace{-4.75 mm} i) For the offline part, $e_{s}$ and $e_{c}$ converge to zero if $w=0$ and to a bounded region around zero if $w \neq 0$.

\hspace{-5 mm} ii) For the online part, the safety constraint \eqref{CBF condition2} is guaranteed for the real system.

\hspace{-5.25 mm} iii) For the online part, the appeared performance loss caused by $e_{c}$, $e_{s}$, and $w$ is mitigated.
\end{problem}


\begin{assumption} [Bounded Terms]
\label{Bounded Terms}
1) $f(x,u)$ is bounded for bounded inputs, 2) There exists $\varepsilon_w> 0$ such that $\|w(k)\| \leq \varepsilon_w$, and 3) $\eta$ is the Lipschitz constant of $h(x)$,i.e.,
\begin{equation}
  \begin{aligned}
    \label{DT Lipschitz constant}
  \left| h(x)-h(\hat{x}) \right| \leq \eta \left\| x-\hat{x} \right\|.  
  \end{aligned}
\end{equation}
\end{assumption}

\begin{figure}[!h]
     \centering
     \includegraphics[width=8.2cm, height=3.9cm]{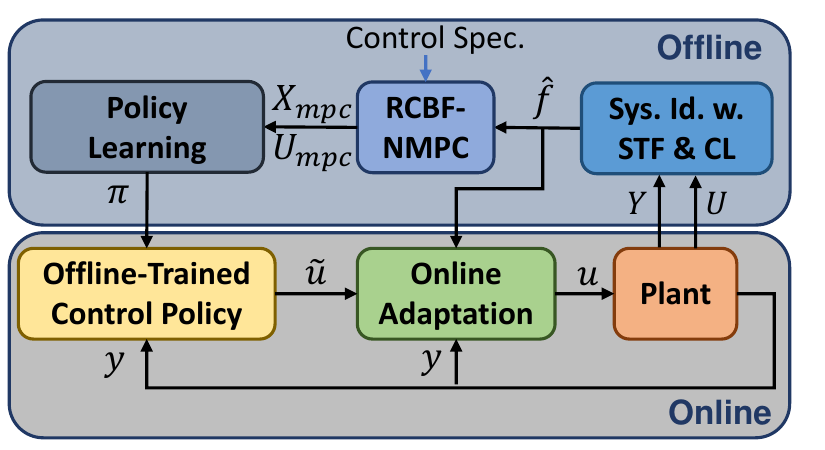}
     \caption{Sequence of steps of online safe data-driven predictive control.}
     \label{Flowchart}
\end{figure}

\section{Main Results}
\label{Sec3}
The developed safe data-driven predictive control pipeline is delineated in Fig.~\ref{Flowchart}. Commencing with amalgamation of discrete-time concurrent learning and spatial temporal filter (STF) \cite{chen2020online}, the primary objective is to attain identification of nonlinear systems under relaxed PE condition. This system identification is subsequently used for implementation of robust CBF (RCBF)-based NMPC to optimize the closed-loop performance while ensuring the system safety against the model uncertainties. Moreover, to enable a computationally efficient online deployment, a policy approximation is trained to emulate the RCBF-based NMPC policy, further aided by an online policy correction that curbs performance deterioration resulting from the existing model uncertainties. The comprehensive constituents of this framework are elaborated upon in the subsequent subsections.

\subsection{STF-based Concurrent Learning}
In this subsection, we present a nonlinear system identification to train the identified model \eqref{nominal model} and enable the NMPC design. Specifically, we identify a nonlinear autoregressive exogenous model (NARX) using the input-output data from the real system \eqref{system} as follows:
\begin{equation}
  \begin{aligned}
    \label{NARX Model}
    & y(k) = G(U_{d}(k),Y_{d}(k)), \\ 
    & U_{d}(k) = [u^{T}(k-1),u^{T}(k-2),\ldots ,u^{T}(k-d_{u})], \\ 
    & Y_{d}(k) = [y^{T}(k-1),y^{T}(k-2),\ldots ,y^{T}(k-d_{y})].
  \end{aligned}
\end{equation}
where $d_{u}$ signifies input delay, $d_{y}$ indicates output delay, and $G$ represents a nonlinear prediction function.

Towards that end, we follow our prior work on STF-based system identification that leverages evolving clustering and recursive least squares (RLS) to systematically decompose the nonlinear system into multiple local models and simultaneously identify the validity zone of each local model   \cite{chen2020online}. Specifically, let $U_{stf}(k) = [U_{d}(k),Y_{d}(k)]^{T}$, the nonlinear model \eqref{NARX Model} is expressed as a composite local model structure, where each local model has a certain valid operating regime, in the following form:
\begin{equation}
\begin{aligned}
\label{Composite Structure}
& y(k) = F(U_{stf}(k),\omega(k);\Phi ,\Psi) \\
& \hspace{6.6 mm}    =\sum\limits_{i=1}^{L} {\alpha }_{i}({U}_{stf}(k);\phi_i,\psi_i) {f}_{i}({U}_{stf}(k),\omega_i(k);{\phi }_{i}),
\end{aligned}
\end{equation}
where ${f}_{i}({U}_{stf}(k),\omega_i(k);{\phi }_{i})$ corresponds to the $i$th local model and can be defined as neural network, Markov chain, linear model, a point, etc. The local models are parameterized by $\phi_{i}$ and account for unstructured uncertainty $\omega_i(k)$. The weighting functions ${\alpha }_{i}({U}_{stf}(k);\phi_i,\psi_i)$ facilitate local model interpolation and are parameterized by $\psi_{i}$. The weighting functions are based on a dissimilarity metric that combines clustering and local model prediction errors. $L$ is the number of local models, and  $\Phi=[{\phi }_{1},{\phi }_{2},\ldots,{\phi }_{L}]$, $\Psi=[{\psi }_{1},{\psi }_{2},\ldots,{\psi }_{L}]$, and $\omega(k)=[w_{1}(k),w_{2}(k),\ldots,w_{L}(k)]$ are the collection of local model parameters, local interpolating function parameters, and local unknown disturbances, respectively.

In particular, we consider linear local models and a softmax-like interpolation function as: 
\begin{equation}
  \begin{aligned}
    \label{local models}
  &{f}_{i}({U}_{stf}(k),\omega_i(k);{\phi }_{i}) = {A}_{i}{U}_{stf}(k)+{b}_{i}+\omega_i(k), \\ 
  &{\alpha }_{i}({U}_{stf}(k);\phi_i,\psi_i) = 
  \frac{\exp(-D_i({U}_{stf}(k);\phi_i,\psi_i))}{\sum \limits_{j=1}^{L}{\exp (-D_j({U}_{stf}(k);\phi_i,\psi_i))}}.
  \end{aligned}
\end{equation}
$A_i \in \mathbb{R}^{n_y \times n_U}$ and $b_i \in \mathbb{R}^{n_y}$ represent the local model parameters $\phi_{i}$, and $\omega_i(k) \in \mathbb{R}^{n_y}$ denotes the unstructured uncertainty (e.g. unknown disturbance) for each local model. Moreover, ${\alpha }_{i}({U}_{stf}(k);\phi_i,\psi_i)$ represents the softmax function, and $D_i({U}_{stf}(k);\phi_i,\psi_i)$ is the dissimilarity metric, including the model residual and a Mahalanobis distance \cite{chen2020online}. The STF capitalizes on evolving clusters with ellipsoidal forms as the foundation for local model interpolation. Each ${f}_{i}$ is associated with an evolving cluster such that the clusters and the local model parameters are updated simultaneously. The readers are referred to \cite{chen2020online} for more details about the STF.\footnote{A video of the simulation result on the STF-Idnetifier can be found online at https://www.youtube.com/watch?v=UYZiNC1LJwM\&t=12s.}

The standard STF uses RLS to update the local model parameters $\{A_i,b_i\}$, which requires fulfilling the persistence of excitation (PE) condition for the system identification. However, the PE condition inevitably increases the complexity and may not be applicable in systems where the control inputs are not completely programmable. Therefore, we propose a discrete-time STF-based concurrent learning to relax the PE condition using a rank condition that is convenient for online inspection and implementations.  

Specifically, using \eqref{Composite Structure} and \eqref{local models}, one has
\begin{equation}
  \begin{aligned}
    \label{Composite Structure b}
  &y(k)=\sum\limits_{i=1}^{L} {\alpha }_{i}({U}_{stf}(k);{\psi }_{i}) ({A}_{i}{U}_{stf}(k)+{b}_{i}+\omega_i(k)),
  \end{aligned}
\end{equation}
where can be rewritten in the following form:
\begin{equation}
  \begin{aligned}
    \label{Regressor Format}
  &y(k)=\Phi \zeta({U}_{stf}(k),\alpha (k))+\omega(k), 
  \end{aligned}
\end{equation}
where $ \Phi = \left[A_1, b_1,\ldots, A_L, b_L  \right]\in \mathbb{R}^{n_y \times q}$, $\zeta({U}_{stf},\alpha) = \left[\alpha_1 {U}^{T}_{stf},\, \alpha_1, \ldots, \alpha_L {U}^{T}_{stf},\, \alpha_L \right]^T\in \mathbb{R}^{q}$, and $q=L(n_U+1)$. Moreover, the total unknown disturbance $\omega$ represents both approximation error, due to the simplification of the original nonlinear system (1) as a composite local model structure (17), and the unknown external disturbance $w(k)$ in the original nonlinear system (1). We assumed a bounded external disturbance $w(k)$ in Assumption 1, and to make a bounded approximation error, one needs to consider a large enough value for the number of local models $L$. $\omega$ is given as $\omega=\alpha_1 \omega_1+ \ldots+ \alpha_L \omega_L \in \mathbb{R}^{n_y}$, where $\|\omega(k)\| \leq \varepsilon$ and $\varepsilon=\alpha_1 \varepsilon_1+ \ldots+ \alpha_L \varepsilon_L$ with $\varepsilon_i$ being the bound of $\omega_i$, i.e.,  $\|\omega_i(k)\| \leq \varepsilon_i$.

Therefore, the identified model is
\begin{equation}
  \begin{aligned}
    \label{Regressor Identified Model}
  &\hat{y}(k)=\hat{\Phi}(k) \zeta({U}_{stf}(k),\alpha (k)), 
  \end{aligned}
\end{equation}
where $\hat{\Phi} = \left[\hat{A}_1, \hat{b}_1, \ldots, \hat{A}_L, \hat{b}_L  \right]$. Now, the system identification error is given as:
\begin{equation}
  \begin{aligned}
    \label{System Identification Error}
  &{e}_{si}(k) = \hat{y}(k) - y(k)\\
  & \hspace{8.75 mm}= \tilde{\Phi}(k) \zeta({U}_{stf}(k),\alpha (k))-\omega(k),
  \end{aligned}
\end{equation}
where $\tilde{\Phi}(k) = \hat{\Phi}(k) - \Phi$ denotes the parameter identification errors. Note that since the output $y$ is measurable, we can measure the system identification error ${e}_{si}$.

The concurrent learning uses a memory of past data as:
\begin{equation}
  \begin{aligned}
    \label{Recorded Data}
  &Z = \left[\zeta({U}_{stf}(k_1),\alpha (k_1)), \ldots, \zeta({U}_{stf}(k_s),\alpha (k_s)) \right],
  \end{aligned}
\end{equation}
where $k_1, \ldots, k_s$ represent the historical time steps of recorded data in the past, and $s$ signifies the number of stored data. For the present time step $k$, the system identification error corresponding to the $j$th sample $\zeta({U}_{stf}(k_j),\alpha (k_j))$ is denoted as follows:
\begin{equation}
  \begin{aligned}
    \label{System Identification Error j}
  &{e}_{si}(k_j) = \tilde{\Phi}(k) \zeta({U}_{stf}(k_j),\alpha (k_j))-\omega(k_j), \hspace{1 mm} j = 1, 2, \ldots, s,  
  \end{aligned}
\end{equation}
where $\tilde{\Phi}(k)$ represents the parameter identification error at the present time step. Then, using a normalizing signal $m(k)=\sqrt[]{\varrho+\zeta(k)^T \zeta(k)}, \varrho>0$, one has
\begin{equation}
  \begin{aligned}
    \label{Normalized System Identification Error}
  \bar{e}_{si}(k) &= \tilde{\Phi}(k) \bar{\zeta}({U}_{stf}(k),\alpha (k))-\bar{\omega}(k),\\
  \bar{e}_{si}(k_j) &= \tilde{\Phi}(k) \bar{\zeta}({U}_{stf}(k_j),\alpha (k_j))-\bar{\omega}(k_j),
  \end{aligned}
\end{equation}
where $\bar{\zeta}(k)=\frac{\zeta(k)}{m(k)}$, $\bar{\omega}(k)=\frac{\omega(k)}{m(k)}$, $\bar{\zeta}(k_j)=\frac{\zeta(k_j)}{m(k_j)}$, and $\bar{\omega}(k_j)=\frac{\omega(k_j)}{m(k_j)}$.

\begin{condition} [Rank Condition]
\label{Rank Condition}
The rank of $Z$ \eqref{Recorded Data} coincides with the dimension of $\zeta({U}_{stf},\alpha)$; i.e., $rank(Z)=q$.
\end{condition}

Using Condition \ref{Rank Condition}, it is obvious that the following inequalities hold:
\begin{equation}
  \begin{aligned}
    \label{Positive Matrix}
  H_1=\sum\limits_{j=1}^{s} \bar{\zeta}({U}_{stf}(k_j),\alpha (k_j)) \bar{\zeta}^T({U}_{stf}(k_j),\alpha (k_j)) > 0, 
  \end{aligned}
\end{equation}
\begin{equation}
  \begin{aligned}
    \label{Positive Matrix 2}
  H_2=\bar{\zeta}({U}_{stf}(k),\alpha (k)) \bar{\zeta}^T({U}_{stf}(k),\alpha (k)) + H_1 >0. 
  \end{aligned}
\end{equation}

Now, using \eqref{Positive Matrix} and \eqref{Positive Matrix 2}, the STF-based concurrent learning law is proposed for the nonlinear system identification in the following theorem.

\begin{theorem} [Discrete-Time STF-based Concurrent Learning]
\label{Discrete-Time Concurrent Learning}
Suppose Condition \ref{Rank Condition} is satisfied. Consider the nonlinear system \eqref{Regressor Format} and the identified model \eqref{Regressor Identified Model}. Then, the discrete-time concurrent learning law
\begin{equation}
  \begin{aligned}
    \label{Update Law}
  &\hat{\Phi}(k+1)=\hat{\Phi}(k) - \bar{e}_{si}(k) \bar{\zeta}^T({U}_{stf}(k),\alpha (k)) \hspace{0.5 mm} \Omega\\
  & \hspace{16 mm} - \sum\limits_{j=1}^{s} \bar{e}_{si}(k_j) \bar{\zeta}^T({U}_{stf}(k_j),\alpha (k_j)) \hspace{0.5 mm} \Omega
  \end{aligned}
\end{equation}
with the learning rate matrix $\Omega=r I_q$,
\begin{equation}
  \begin{aligned}
    \label{Learning rate}
  &0<r<\frac{2\lambda_{min}(H_2)}{\lambda_{max}^2(H_2)}
  \end{aligned}
\end{equation}
guarantees that\\
i) $e_{si}$ converges to zero when $\omega(k)=0$.\\
ii) $e_{si}$ is uniformly ultimately bounded (UUB) when $\omega(k) \neq 0$.
\end{theorem}
\begin{proof}
Consider the subsequent Lyapunov function candidate:
\begin{equation}
  \begin{aligned}
    \label{Lyapunov Candidate}
  &V(k) = tr \{\tilde{\Phi}(k) \Omega^{-1} \tilde{\Phi}^T(k)\},  
  \end{aligned}
\end{equation}
where $\tilde{\Phi}(k)$ and $\Omega=r I_q$ are the parameter identification errors and the learning rate matrix defined in \eqref{System Identification Error} and \eqref{Learning rate}, respectively. Therefore, one has
\begin{equation}
  \begin{aligned}
    \label{Lyapunov Candidate Derivative}
  &V(k+1)-V(k)\\ 
  &= tr\{\tilde{\Phi}(k+1) \Omega^{-1} \tilde{\Phi}^T(k+1)- \tilde{\Phi}(k) \Omega^{-1} \tilde{\Phi}^T(k)\}\\ 
  &= tr\{(\tilde{\Phi}(k+1)-\tilde{\Phi}(k)) \Omega^{-1} (\tilde{\Phi}(k+1)+\tilde{\Phi}(k))^T\}\\
  &=tr\{(-\bar{e}_{si}(k) \bar{\zeta}^T(k) \hspace{0.5 mm} \Omega -\sum\limits_{j=1}^{s} \bar{e}_{si}(k_j) \bar{\zeta}^T(k_j) \hspace{0.5 mm} \Omega)\Omega^{-1}\\
  &\hspace{4.5 mm}(-\bar{e}_{si}(k) \bar{\zeta}^T(k) \hspace{0.5 mm} \Omega -\sum\limits_{j=1}^{s} \bar{e}_{si}(k_j) \bar{\zeta}^T(k_j) \hspace{0.5 mm} \Omega + 2\tilde{\Phi}(k))^T\}\\
  &=tr\{(-\tilde{\Phi}(k) H_2+H_3) \hspace{0.5 mm} (-\tilde{\Phi}(k) H_2\Omega+H_3\Omega+ 2\tilde{\Phi}(k))^T\},\\
  \end{aligned}
\end{equation}
where $H_3=\bar{\omega}(k)\bar{\zeta}^T(k)+\sum\limits_{j=1}^{s} \bar{\omega}(k_j) \bar{\zeta}^T(k_j)$, and $\|H_3\| \leq \bar{\varepsilon}_n$. Here, one can obtain $\bar{\varepsilon}_n$ using the bound of $\omega$. 
Now, one has
\begin{equation}
  \begin{aligned}
    \label{Lyapunov Candidate Derivative 2}
  &V(k+1)-V(k)\\ 
  &=tr\{\tilde{\Phi}(k) P_1 \tilde{\Phi}^T(k)+\tilde{\Phi}(k) P_2+P_3 \tilde{\Phi}^T(k)+P_4\},\\
  \end{aligned}
\end{equation}
where $P_1=H_2 \Omega^T H^T_2-2H_2$, $P_2=-H_2\Omega^TH^T_3$, $P_3=-H_3 \Omega^T H^T_2+2H_3$, and $P_4=H_3\Omega^TH^T_3$.
Therefore, one can show
\begin{equation}
  \begin{aligned}
    \label{Lyapunov Candidate Derivative 3}
  &V(k+1)-V(k) \leq Q_1 \|\tilde{\Phi}(k)\|^2 + Q_2 \|\tilde{\Phi}(k)\| + Q_3,\\
  \end{aligned}
\end{equation}
where $Q_1=r \lambda^2_{max}(H_2)-2\lambda_{min}(H_2)$, $Q_2=r\lambda_{min}(H_2)\bar{\varepsilon}_n$ $+r\lambda_{min}(H_2)\bar{\varepsilon}_n+2\bar{\varepsilon}_n$, and $Q_3=r\bar{\varepsilon}^2_n$. Now, using \eqref{Learning rate}, it is clear that $Q_1 < 0$; therefore, using \eqref{Lyapunov Candidate}, one has the following inequality when $\omega(k)=0$:
\begin{equation}
    \label{Lyapunov Candidate Derivative 4}
  V(k+1)-V(k) \leq Q_1 \|\tilde{\Phi}(k)\|^2\leq r \hspace{0.5 mm} Q_1 V(k)< 0.
\end{equation}
where it shows that $\tilde{\Phi}$ converges to zero; thus, $e_{si}$ converges to zero according to \eqref{System Identification Error}. This completes the proof of the first part.

\hspace{-3.55 mm}Now, for $\omega(k) \neq 0$, since $\|\tilde{\Phi}(k)\| \geq 0$, $Q_1 < 0$, $Q_2 \geq 0$, and $Q_3 \geq 0$, the only valid non-negative root of \eqref{Lyapunov Candidate Derivative 3} is
\begin{equation}
  \begin{aligned}
    \label{parameter bound}
  &\tilde{\Phi}_b=\frac{-Q_2-\sqrt{Q^2_2-4 Q_1 Q_3}}{2Q_1}.\\
  \end{aligned}
\end{equation}
Thus, when $\|\tilde{\Phi}(k)\| > \tilde{\Phi}_b$, one has
\begin{equation}
  \begin{aligned}
    \label{Lyapunov Candidate Derivative 5}
  &V(k+1)-V(k) < 0,  
  \end{aligned}
\end{equation}
which makes $\tilde{\Phi}(k)$ to enter and stay in the compact set $S_{\tilde{\Phi}}=\{\tilde{\Phi}: \|\tilde{\Phi}\| \leq \tilde{\Phi}_b\}$; therefore, one can conclude that $e_{si}$ converges to a small region around zero using \eqref{System Identification Error}. This completes the proof of the second part.
\end{proof}

\begin{figure}[!h]
     \centering
     \includegraphics[width=1\linewidth]{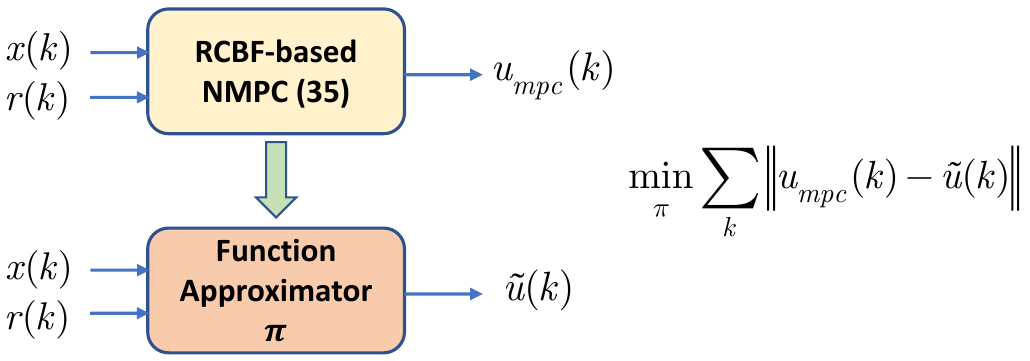}
     \caption{Schematic of NMPC policy function approximation. 
     }
     \label{fig:policyLearning}
\end{figure}

\begin{remark} [System Identification]
\label{System Identification}
Theorem 1 presents the learning law of the proposed STF-based concurrent learning. We investigate two cases: i) the real system without the external disturbance and the simplification error, and ii) the real system with the bounded external disturbance and the bounded simplification error. For the first case, i.e., $\omega(k) = 0$, we prove that $e_{si}$ converges to zero, and the real system is identified perfectly. For the second case with $\omega(k) \neq 0$, although it is known that $e_{si}$ does not converge to zero, we prove that it is UUB, i.e., $e_{si}$ converges to a small region around zero \eqref{parameter bound}.
\end{remark}

In the next subsection, we use the developed STF-based concurrent learning once again to learn the control policy.

\subsection{Safe Data-Driven Predictive Control}
Using the trained STF model that describes the input–output relationship, Appendix A presents the transformation of the STF model to a state-space model, i.e., the identified model (5), by using the past inputs and outputs as the states of the model so that one can implement the NMPC (6) and the CBF (9). However, the NMPC is computationally expensive; therefore, we develop a safe data-driven predictive control using a STF-based policy approximator. As shown in Fig.~\ref{fig:policyLearning}, we utilize the STF-based concurrent learning to approximate the NMPC policy \eqref{NMPC}, which can be viewed as a mapping from the state $x(k)$ and the reference $r(k:k+N)$ to the optimal control $u_{mpc}(k)$. Using the trained policy approximator, we only have simple algebraic computations; thus, the online control computation is greatly reduced compared to the original optimization problem \eqref{NMPC}. Note that the STF-based state-space model is essentially an input-output model with the state being past inputs and outputs that are available, without requiring measured ``system states'' in a typical state-space model.

 Due to the model uncertainties, the safety constraints may not be satisfied for the real system. Therefore, we develop the RCBF-based NMPC so that the approximated control policy $\tilde{u}(k)$ guarantees the system safety against $e_{c}$, $e_{s}$, and $w$. Specifically, according to Remark \ref{System Identification}, one knows that the drift dynamics errors \eqref{identification error} are bounded, i.e., $\|e_{s}(k)\| \leq \varepsilon_s$ and $\|e_{c}(k)\| \leq \varepsilon_c$, where $\varepsilon_s$ and $\varepsilon_c$ are obtained using Theorem \ref{Discrete-Time Concurrent Learning} and the necessary derivations. Now, the subsequent theorem introduces a robust constraint satisfaction using the identified model \eqref{nominal model} to ensure the safety of the nonlinear system \eqref{disturbed system}.

\begin{theorem} [Robust Constraint Satisfaction]
\label{Robust Control Barrier Function}
Consider the nonlinear system \eqref{disturbed system}, the identified model \eqref{nominal model}, and the set \eqref{Safe set}. Defining the RCBF $h_r(\hat{x}) = h(\hat{x}) +\eta (\varepsilon_w + \varepsilon_s + \varepsilon_c)$, the robust safety constraint
\begin{equation}
  \begin{aligned}
    \label{RCBF Constraint}
  &h(\hat{x}(k+1)) \leq (1-\gamma)h(\hat{x}(k)) -\gamma \eta (\varepsilon_w + \varepsilon_s + \varepsilon_c)  
  \end{aligned}
\end{equation}
guarantees the forward invariance of the set $S$ against the model uncertainties $w$, $e_{s}$, and $e_{c}$.
\end{theorem}

\begin{proof}
For simplicity of notation, we represent $x(k)$ as $x_k$ in the subsequent equation. The mean value theorem along with the nonlinear system \eqref{disturbed system} and the identified model \eqref{nominal model} yields
\begin{equation}
  \begin{aligned}
    \label{Mean Value}
  &h(x_k) = h(\hat{x}_k + (x_k-\hat{x}_k)) \\
  & \hspace{8.5 mm}= h(\hat{x}_k) + (h(x_k)-h(\hat{x}_k))\\
  & \hspace{8.5 mm} \leq h(\hat{x}_k) + \left| h(x_k)-h(\hat{x}_k) \right|\\
  & \hspace{8.5 mm} \leq h(\hat{x}_k) + \eta \left\| x_k-\hat{x}_k \right\| \\
  & \hspace{8.5 mm} \leq h(\hat{x}_k) + \eta \left( \left\| w_{k-1} \right\| + \left\| {e_{fs}}_{k-1} \right\| + \left\| {e_{fc}}_{k-1} \right\| \right) \\
  & \hspace{8.5 mm} \leq h(\hat{x}_k) + \eta (\varepsilon_w + \varepsilon_s + \varepsilon_c)\\
  & \hspace{8.5 mm} = h_r(\hat{x}_k).
  \end{aligned}
\end{equation}
Thus, to guarantee $h(x_k) \leq 0$, one can ensure $h_r(\hat{x}_k) \leq 0$. Using \eqref{CBF condition}, the safety constraint on $h_r(\hat{x}_k)$ is expressed as:
\begin{equation}
\begin{aligned}
    \label{Robust CBF condition}
  &\triangle h_r(\hat{x}(k)) \leq -\gamma h_r(\hat{x}(k)).
  \end{aligned}
\end{equation}
where guarantees $h_r(\hat{x}(k)) \leq 0$.
Now, substituting \eqref{Mean Value} into \eqref{Robust CBF condition} yields
\begin{equation}
\begin{aligned}
    \label{Robust CBF condition2}
  &h(\hat{x}(k+1)) - h(\hat{x}(k)) \leq - \gamma \left( h(\hat{x}(k)) +\eta (\varepsilon_w + \varepsilon_s + \varepsilon_c) \right),
  \end{aligned}
\end{equation}
which is \eqref{RCBF Constraint}. This completes the proof.
\end{proof}

Now, considering $X=\{x \in \mathbb{R}^n: h(x) \leq 0\}$ as the state constraint, we modify the NMPC \eqref{NMPC} with the following RCBF-based NMPC design:
\begin{equation}
  \begin{aligned}
    \label{RCBF-based NMPC}
    &(\mathbf{x}^{*},\mathbf{u}^{*},\mathbf{\gamma}^{*}) = \underset{\mathbf{\hat{x}},\mathbf{u},\mathbf{\gamma}}{\arg\min} \hspace{1 mm} J_N(\mathbf{\hat{x}},\mathbf{u}) + \varphi(\gamma)\\
    &s.t.\quad \hat{x}(k+1) = \hat{f}(\hat{x}(k),u(k)),\quad \hat{y}(k) = g(\hat{x}(k)),\\
    & \hspace{8.5 mm} \hat{x}(0) = \hat{x}_0,\quad u(k) \in U,\\
    & \hspace{8.5 mm} h(\hat{x}(k+1)) \leq (1-\gamma)h(\hat{x}(k)) -\gamma \eta (\varepsilon_w + \varepsilon_s + \varepsilon_c), \\
    & \hspace{8.5 mm} 0 < \gamma \leq 1,
  \end{aligned}
\end{equation}
where $\varphi(\gamma) = P \gamma^2$ represents a regularization term applied to the optimization variable $\gamma$, and $P > 0$. It is worth noting that the feasibility of the optimization problem \eqref{RCBF-based NMPC} guarantees that the control input satisfies the true CBF constraint (9). The optimization problem \eqref{RCBF-based NMPC} is formulated with the identified model \eqref{nominal model}; therefore, $\epsilon_w$ and $\epsilon_s$ are added to the CBF constraint of the NMPC so that the RCBF guarantees the system safety under the model uncertainties. On the other hand, since we learn the RCBF-based NMPC policy as $\tilde{u}$, it may violate the safety of the real system. Thus, one needs to make more conservative CBF constraint for the NMPC by considering the control policy learning error $\epsilon_c$.

\begin{remark} [Feasibility]
\label{Robust Safety}
The feasibility of the NMPC with both distance constraint ($\gamma = 1$) and input constraint is a challenging problem. Moreover, changing the distance constraint to the CBF constraint makes the feasibility of the NMPC problem \eqref{RCBF-based NMPC} more challenging since it considers more conservative constraint to improve the system safety. In the optimization problem \eqref{RCBF-based NMPC}, if $\gamma$ becomes relatively small, the sublevel set of the RCBF $h$ will be smaller, and the system tends to be safer; however, the intersection between the reachable set and the sublevel set might be infeasible. When $\gamma$ becomes larger, the sublevel set will be increased in the state space, which makes the optimization problem more likely to be feasible; however, the RCBF constraint might not be active during the optimization. If $\gamma = 1$, the RCBF constraint describes a simple distance constraint which requires the NMPC with a long horizon to ensure the system safety for the real system. Consequently, we adjust the decay rate of the CBF, i.e., $\gamma$, from a fixed constant value to an optimization variable so that it improves the feasibility of the optimization problem under the CBF constraint. Assuming the optimization problem \eqref{RCBF-based NMPC} is feasible under a distance constraint, the regularization term $\varphi(\gamma)$ makes the optimization problem feasible under CBF constraint, which the considered value for $P$ controls the trade off between improving safety (making $\gamma$ close to $0$) and keeping feasibility (keeping $\gamma$ close to $1$). It is worth noting that opting for a relatively too small value of $P$ is not advisable, as it could lead to an excessive relaxation of the RCBF constraint and make the optimized value $\gamma$ closer to 1. Formal guarantee on the recursive feasibility of the RCBF-based NMPC problem \eqref{RCBF-based NMPC} and the stability of the closed-loop system \cite{rakovic2022homothetic} requires more analysis and will be addressed in our future work.
\end{remark}

\subsection{Online Adaptive Control Policy}
Although the approximated control policy $\tilde{u}(k)$ may already lead to a reasonable closed-loop performance for many cases, one can see a performance loss for a system due to two main reasons i) the selection of the hyperparameters governing the architecture of the STF-based concurrent learning, and ii) limited training data for some regions of the feasible state space. To minimize the performance loss caused by $e_{c}$, $e_{s}$, and $w$, it is desirable to correct $\tilde{u}(k)$ online. 


\begin{corollary}[Steady-State Optimization Problem \cite{krishnamoorthy2021adaptive}]
\label{lemm1}
Using the RCBF-based NMPC policy, a unique equilibrium pair $(x_e,u_e)$ minimizes the steady-state optimization problem
\begin{equation}
  \begin{aligned}
    \label{Steady-State NMPC}
    &(\mathbf{x}_{e},\mathbf{u}_{e},\mathbf{\gamma}_{e}) = \underset{\mathbf{\hat{x}},\mathbf{u},\mathbf{\gamma}}{\arg\min} \hspace{1 mm} l(\hat{x},u,\gamma)\\
    &s.t.\quad \hat{x} = \hat{f}(\hat{x},u),\quad \hat{y} = g(\hat{x})\\
    & \hspace{8.5 mm} u \in U,\\
    & \hspace{8.5 mm} h(\hat{x}) \leq (1-\gamma)h(\hat{x}) -\gamma \eta (\varepsilon_w + \varepsilon_s + \varepsilon_c),
  \end{aligned}
\end{equation}
where $l(\hat{x},u,\gamma) = \phi(\hat{x},u,\hat{y},r) + \psi(\hat{x},\hat{y},r) + \varphi(\gamma)$ represents the steady-state cost function. \hspace{42 mm} $\square$
\end{corollary}

\begin{condition}[Steady-State Approximated Optimal Control]
\label{cond1}
$(x^{\prime}_e,u^{\prime}_e)$ is the asymptotically stable equilibrium point of the nominal model \eqref{nominal model} under the approximated optimal control policy $\tilde{u}$.
\end{condition}

\begin{lemma}[Modified Steady-State Optimization Problem]
\label{Modified Steady-State Optimization Problem}
For the equilibrium point $(x_e,u_e)$, the steady-state optimization problem \eqref{Steady-State NMPC} is implicitly expressed in the following form:
\begin{equation}
  \begin{aligned}
    \label{Modified Steady-State NMPC}
    &\mathbf{u}_{e} = \underset{\mathbf{u}}{\arg\min} \hspace{1 mm} \tilde{l}(u)\\
    &s.t.\quad \tilde{h}(u)\leq 0,
  \end{aligned}
\end{equation}
where the constraints of the optimization problem \eqref{Steady-State NMPC} are collectively denoted as $\tilde{h}(u)$.
\end{lemma}

\begin{proof}
Let $\hat{x}(k)$ and $\tilde{u}(k)$ denote the nominal trajectory using the function-approxiamted control policy. Thus, the nominal trajectory for the NMPC policy \eqref{RCBF-based NMPC} is expressed as
\begin{equation}
  \begin{aligned}
    \label{Difference}
    & \hat{x}_{mpc}(k)=\hat{x}(k)+\delta \hat{x}(k),\\
    & u_{mpc}(k)=\tilde{u}(k)+\delta u(k),
  \end{aligned}
\end{equation}
where $\delta u(k)$ represents the difference between the NMPC and approximated control policy at each time step $k$, and $\delta \hat{x}(k)$ denotes the resulting change on the system states. Using the nominal model \eqref{nominal model}, one has
\begin{equation}
  \begin{aligned}
    \label{Variation}
    & \delta \hat{x}(k+1)= \hat{f}_{\hat{x}} (k) \delta \hat{x}(k) + \hat{f}_{u} (k) \delta u(k),\\
    & \delta l(k)= l_{\hat{x}}(k) \delta \hat{x}(k) + l_{u} (k) \delta u(k),
  \end{aligned}
\end{equation}
where $\hat{f}_{\hat{x}}(k)$, $\hat{f}_{u}(k)$, $l_{\hat{x}}(k)$, and $l_{u}(k)$ are partial derivatives with respect to $\hat{x}$ and $u$, respectively, and  obtained using $\hat{x}(k)$ and $\tilde{u}(k)$.

\hspace{-5 mm} Now, it is clear that $\delta \hat{x}(k+1)=\delta \hat{x}(k)$ for the steady-state condition; therefore, one has
\begin{equation}
  \begin{aligned}
    \label{eq13}
    \delta \hat{x}(k)= (I_n-\hat{f}_{\hat{x}}(k))^{-1} \hat{f}_{u}(k) \hspace{0.75 mm} \delta u(k),
  \end{aligned}
\end{equation}
and
\begin{equation}
  \begin{aligned}
    \label{eq14}
    \delta l(k)= (l_{\hat{x}}(k) (I_n-\hat{f}_{\hat{x}}(k))^{-1} \hat{f}_{u}(k) + l_{u} (k)) \hspace{0.75 mm} \delta u(k).
  \end{aligned}
\end{equation}
Using \eqref{eq14}, one has $\delta \tilde{l}(u) = (l_{\hat{x}} (I_n-\hat{f}_{\hat{x}})^{-1} \hat{f}_{u} + l_{u})$ and can obtain $\tilde{l}(u)$. Similarly, $\tilde{h}(u)$ is obtained using the same process. This completes the proof.
\end{proof}

Now, consider $\tilde{x}=\hat{x}-x$ and define an auxiliary variable $s$ as
\begin{equation}
  \begin{aligned}
    \label{s}
    s(k)=\Gamma \tilde{x}(k),
  \end{aligned}
\end{equation}
where $\Gamma \in \mathbb{R}^{m \times n}$ is designed such that $\Gamma \hat{f}_u \in  \mathbb{R}^{m \times m}$ is a diagonal matrix. The following theorem presents the proposed online policy correction scheme, including an KKT adaptation and an ancillary feedback control, to minimize the performance loss caused by the model uncertainties.

\begin{theorem}[Online Control Policy Correction]
\label{Online Adaptive Control Policy}
Considering the real system \eqref{disturbed system}, the nominal model \eqref{nominal model}, Corollary \ref{Modified Steady-State Optimization Problem}, Lemma \ref{Modified Steady-State Optimization Problem}, and the auxiliary variable $s$ \eqref{s}, the online adaptation policy
\begin{equation}
  \begin{aligned}
    \label{Online Adaptation}
    & u(k)=\tilde{u}(k)+\delta u(k)+K(\hat{x}(k),x(k)),\\
    & \delta u(k) \approx \delta u(k-1)-K_0 \begin{bmatrix} \tilde{h}_a(u(k-1))\\ \mathcal{A}^T \delta \tilde{l}(u(k-1)) \end{bmatrix},\\
    & K_0=\begin{bmatrix} \delta \tilde{h}_a(u(k-1))\\ \mathcal{A}^T \delta^2 \tilde{l}(u(k-1)) \end{bmatrix}^{-1},\\
    & K(\hat{x}(k),x(k))=\frac {1}{\Gamma \hat{f}_{u}(\hat{x}(k),\tilde{u}_{mpc}(k))} (\Upsilon s(k)+\Gamma \varepsilon_w+\Gamma \varepsilon_s),
  \end{aligned}
\end{equation}
with $\mathcal{A}^T \delta \tilde{h}_a(u)^T=0$, design matrix $\Upsilon \in \mathbb{R}^{m \times m}$ $(0 \leq \Upsilon_{ii}<1, \hspace{1 mm} i=1,...,m)$, and $\tilde{u}_{mpc}(k) = \tilde{u}(k)+\delta u(k)$, minimizes the performance loss caused by the model uncertainties.
\end{theorem}

\begin{proof}
The first part of the proof focuses on minimizing the performance loss for the nominal model \eqref{nominal model} caused by the control learning error. Considering Corollary \ref{Modified Steady-State Optimization Problem} and Lemma \ref{Modified Steady-State Optimization Problem}, one can conclude that $(x_e,u_e)$ represents the asymptotically stable equilibrium point of the nominal model under the RCBF-based NMPC policy. Hence, $u_{mpc}$ satisfies the KKT conditions of \eqref{Modified Steady-State NMPC}, which are expressed as
\begin{equation}
  \begin{aligned}
    \label{KKT}
    & \delta \tilde{l}(u) + \delta \tilde{h}_a(u)^T \lambda=0,\\
    & \lambda^T \tilde{h}(u)=0,
  \end{aligned}
\end{equation}
where $\tilde{h}_a(u)$ denotes the active constraints, and $\lambda\geq 0$ is the Lagrange multiplier. Now, one can rewrite \eqref{KKT} as
\begin{equation}
  \begin{aligned}
    \label{KKT 2}
    & \mathcal{A}^T \delta \tilde{l}(u)=0,\\
    & \tilde{h}_a(u)=0,
  \end{aligned}
\end{equation}
where $\mathcal{A}$ lies in the null space of the active constraint variation, i.e., $\mathcal{A}^T \delta \tilde{h}_a(u)^T=0$.

\hspace{-5 mm} Due to the control learning error $e_{c}$, one may have $(x^{\prime}_e,u^{\prime}_e) \neq (x_e,u_e)$; therefore, the goal is that the approximated control policy $\tilde{u}(k)$ is adapted online such that it guarantees $(x^{\prime}_e,u^{\prime}_e)=(x_e,u_e)$. Using Lemma \ref{lemm1}, it is clear that $(x^{\prime}_e,u^{\prime}_e)$ does not satisfy the KKT conditions \eqref{KKT 2} if $(x^{\prime}_e,u^{\prime}_e) \neq (x_e,u_e)$. Thus, the deviation from the KKT condition \eqref{KKT 2} signifies an asymptotic performance loss because of the RCBF-based NMPC policy approximation. Consequently, $\delta u$ minimizes the asymptotic performance loss due to $e_{c}$, and the gain $K_0$ is selected such that it fine-tunes asymptotic performance while minimizing its impact on the system dynamics. Therefore, the first part of the proof is completed.

\hspace{-5 mm} Now, we have a reasonable performance for the nominal model under the approximated control policy $\tilde{u}(k)$; however, the performance loss still exists for the real system \eqref{disturbed system} due to the unknown disturbance and the system identification error. Considering the nominal model \eqref{nominal model} under $\tilde{u}_{mpc}$, one has
\begin{equation}
  \begin{aligned}
    \label{x tilde}
    &\tilde{x}(k+1)=\hat{x}(k+1)-x(k+1)\\
    &=\hat{f}(\hat{x}(k),\tilde{u}_{mpc}(k))-f(x(k),u(k))-w(k)\\
    &=\hat{f}(\hat{x}(k),\tilde{u}_{mpc}(k))-\hat{f}(\hat{x}(k),u(k))+e_{s}(k)-w(k)\\
    &= \hat{f}_{u}(\hat{x}(k),\tilde{u}_{mpc}(k)) \hspace{0.7 mm} (\tilde{u}_{mpc}(k)-u(k)) + e_{s}(k) - w(k)\\
    &=- \hat{f}_{u}(\hat{x}(k),\tilde{u}_{mpc}(k)) \hspace{0.7 mm} K(\hat{x}(k),x(k)) + e_{s}(k) - w(k).
  \end{aligned}
\end{equation}
Now, the Lyapunov function candidate is considered as
\begin{equation}
  \begin{aligned}
    \label{Lyapunov}
    V(k)=s(k)^T s(k),
  \end{aligned}
\end{equation}
where $V(k)$ is a positive definite function, and one has
\begin{equation}
  \begin{aligned}
    \label{Lyapunov derivative}
    V(k+1)-V(k)=s(k+1)^T s(k+1)-s(k)^T s(k).
  \end{aligned}
\end{equation}
Using \eqref{s}-\eqref{x tilde}, one can derive
\begin{equation}
  \begin{aligned}
    \label{s 2}
    &s(k+1)=\Gamma \tilde{x}(k+1)\\
    &=-\Upsilon s(k)-\Gamma \varepsilon_w-\Gamma \varepsilon_s+\Gamma e_f(k) -\Gamma w(k)\\
    &\leq -\Upsilon s(k)-\Gamma \varepsilon_w-\Gamma \varepsilon_s+\Gamma \|e_f(k)\| +\Gamma \|w(k)\|\\
    &\leq -\Upsilon s(k).
  \end{aligned}
\end{equation}
Using \eqref{Lyapunov derivative} and \eqref{s 2}, considering $P=I_m - \Upsilon^T \Upsilon$, one has
\begin{equation}
  \begin{aligned}
    \label{Lyapunov derivative 2}
    &V(k+1)-V(k) \leq -s(k)^T P s(k)\\
    &\hspace{25.4 mm} \leq -\lambda_{min}(P)V(k)\\
    &\hspace{25.4 mm} <0.
  \end{aligned}
\end{equation}
Consequently, one can conclude that $\delta u(k)$ minimizes the difference between $u_{mpc}$ and $\tilde{u}$, and $K(\hat{x},x)$ keeps the actual state $x$ of the real system around the nominal trajectory $\hat{x}$ under $\tilde{u}_{mpc}$; therefore, $K(\hat{x},x)$ minimizes the performance loss due to the system identification error $e_{s}$ and the the unknown disturbance $w$. This completes the proof.
\end{proof}

\begin{remark} [KKT Condition]
\label{KKT Condition}
To develop the proposed adaptation law \eqref{Online Adaptation}, Condition 2 must be satisfied for the nominal model, which means that it converges to an equilibrium point $(x^{\prime}_e,u^{\prime}_e)$ under the approximated control policy $\tilde{u}(k)$, but it may not be the desired equilibrium point $(x_e,u_e)$ from the RCBF-based NMPC policy $u_{mpc}(k)$ because of the policy approximation error. When we train the STF-based policy approximator, one wants to accomplish the equilibrium solution $(x_e,u_e)$, where the KKT conditions are satisfied. Consequently, we adapt $\tilde{u}(k)$ online to ensure that $(x^{\prime}_e,u^{\prime}_e)$ holds the KKT conditions and make $(x^{\prime}_e,u^{\prime}_e) = (x_e,u_e)$.
\end{remark}

\begin{remark} [Offline Probabilistic Verification]
\label{Offline Probabilistic Verification}
To guarantee that Condition \ref{cond1} is satisfied, using the offline safe data-driven predictive control, the nominal model is simulated for $N_v$ randomly selected initial states in the training data range as
\begin{equation*}
  \begin{aligned}
    \label{eq20}
    N_v \geq \frac{\log \frac {1}{\kappa }}{\log \frac {1}{1-\epsilon}},
  \end{aligned}
\end{equation*}
where $\epsilon \in (0,1)$ and $\kappa  \in (0,1)$ denote the accuracy and confidence of the offline probabilistic verification. If all $N_v$ closed-loop trajectories are stable ($E=0$), one can conclude that the nominal model under the approximated control policy converges to $(x^{\prime}_e,u^{\prime}_e)$ for all initial states with the probability \cite{krishnamoorthy2021adaptive}
\begin{equation*}
  \begin{aligned}
    \label{eq21}
    Pr\{Pr\{E=0\} \geq 1-\epsilon\} \geq 1-\kappa.
  \end{aligned}
\end{equation*}
However, if any closed-loop trajectory from $N_v$ samples is unstable, the control learning procedure must be repeated.
\end{remark}

\begin{remark} [Computational Cost]
\label{Computational Cost}
The developed safe data-driven predictive control is significantly more computationally efficient compared to the RCBF-based NMPC \eqref{RCBF-based NMPC} since the proposed control framework approximates the NMPC policy while keeping the real system safe; however, it does not need to solve an optimization problem at each time step $k$ (only algebraic computations are needed). In comparison with our previous work \cite{vahidi2022data}, the RCBF is extended to guarantee the system safety against not only the system identification error and the external disturbance but also the control learning error. Therefore, we have removed the QP safety filter from the algorithm. Moreover, the proposed online adaptive control policy enables us to minimize the performance loss; therefore, we do not need to consider a switching criteria for returning the NMPC. These two contributions effectively reduce the computational cost for the proposed algorithm.
\end{remark}

\begin{remark} [Comparison]
\label{Comparison}
In comparison with \cite{chen2020online}, the regressor vector $\zeta({U}_{stf},\alpha)$ has to be persistently exciting, imposing  conditions on past, current, and future regressor vectors that is difficult or impossible to verify online. Instead, Condition 1 only deals with a subset of past data, which makes it easy to monitor. Also, it is convenient to check whether replacing new data will increase $\lambda_{min}(H_2)$ and/or decrease $\lambda_{max}(H_2)$ to reduce the convergence time. Compared to \cite{chowdhary2010concurrent}, the STF-based concurrent learning does not require the measurement or estimation of the derivatives of the system states (1). On the other hand, compared to \cite{taylor2020adaptive,lopez2020robust,garg2021robust,black2021fixed,isaly2021adaptive}, the developed RCBF considers all types of model uncertainties, i.e., the external disturbance, the system identification error, and the control learning error to guarantees the system safety. Moreover, the RCBF and the NMPC are combined to use the advantage of each one, i.e. system safety and optimal control. Last but not least, we improve the performance of the trained controller using the proposed online adaptive control policy, including a KKT adaptation and a feedback control.
\end{remark}

\begin{remark} [Practical Challenges] 
\label{Control Learning Error}
Conditions 1 and 2 are the necessary requirements for the proposed system identification and control policy learning approaches, respectively. Condition 1 represents PE input requirement for system identification as a rank condition on collected data matrix, which is satisfied by recording rich (informative) data. Before satisfying Condition 1, we use RLS to update the local model parameters. On the other hand, Condition 2 illustrates that the performance of the trained control policy is reasonable; thus, if Condition 2 is not satisfied, we must repeat the training procedure with more rich data. In practical control systems, one has no difficulty for satisfying Conditions 1 and 2 if sufficient data is collected to train a model for each part. However, we assume a bounded unknown disturbance in Assumption 1, which affects the bounds on the system identification error and the control learning error. If Assumption 1 is satisfied, Theorems 1, 2, and 3 are robust against the unknown disturbance. The control framework needs to be modified for unbounded unknown disturbances, which will be addressed in our future work. Moreover, for time-varying systems \cite{zhang2016switched}, the proposed framework can easily update the parameters of the trained STF model for system identification using informative real-time data. According to Theorem 1, it is convenient to check whether the new data is informative or does not add any information to the trained model. Using \eqref{Lyapunov Candidate Derivative 3} in Theorem 1, if replacing new data increases $\lambda_{min}(H_2)$ and/or decreases $\lambda_{max}(H_2)$, one should update the trained parameters online. However, for updating the trained STF control policy, one needs to consider an event-triggered NMPC such that the NMPC is returned if the control scheme has performance losses. We will consider this case in our future work.
\end{remark}

\begin{figure}[!h]
     \centering
     \includegraphics[width=0.99\linewidth]{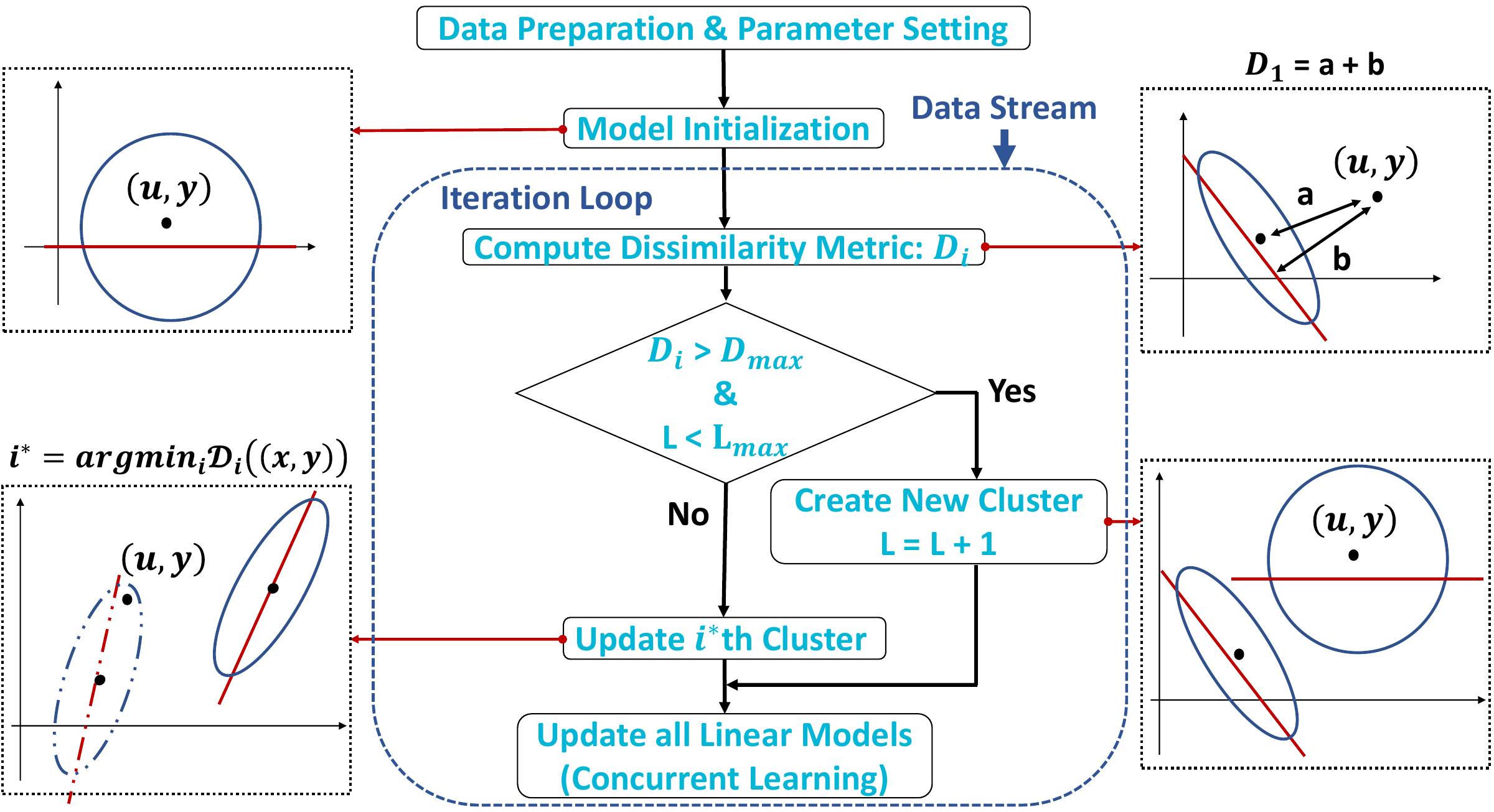}
     \caption{Sequence of steps of discrete-time STF-based Concurrent Learning: i) First cluster-model is initialized for the first I/O data, ii) For next data, the dissimilarity metric is computed to either create a new cluster-model or update one of the existing ones, iii) For each iteration, all linear models are updated using the concurrent learning, and iv) The process is repeated until all training I/O data is used. }
     \label{STF Flowchart}
\end{figure}

Fig. \ref{STF Flowchart} presents a flowchart that describes the steps of the discrete-time STF-based concurrent learning, and the readers are referred to see Algorithms 1 and 2 in \cite{chen2020online} for more details about the sequence of the steps of the STF framework. It is worth noting that we have modified the Algorithm 2 in \cite{chen2020online} such that we use the concurrent learning instead of RLS to remove the PE requirement.

\section{Simulation Results}
\label{Sec4}
In this section, we illustrate the efficacy of the developed safe data-driven predictive control through two simulation examples, i.e., a cart-inverted pendulum and a gasoline engine controls. The cart-inverted pendulum has a known model, where we use it for the NMPC to compare its result with the safe data-driven predictive control. However, there is not any known model for the gasoline engine vehicle; thus, we collect data from the system to identify a nominal model for the NMPC and compare its result with the proposed controller.

\subsection{Cart-Inverted Pendulum}
The cart-inverted pendulum is modeled as (see Fig. \ref{Cart-inverted pendulum fig}):
\begin{equation}
  \begin{aligned}
    \label{Cart-inverted pendulum}
    &\ddot{z}=\frac{F-{K}_{d}\dot{z}-{m}_{pend}(L{\dot{\theta }}^{2} \sin (\theta )-g\sin (\theta ) \cos (\theta ))}{{m}_{cart}+{m}_{pend}{\sin }^{2}(\theta )},\\ 
    &\ddot{\theta }=\frac{g\sin (\theta ) + \ddot{z}\cos (\theta )}{L},
  \end{aligned}
\end{equation}
where $z$ is the cart position, and $\theta$ represents the pendulum angle. ${m}_{pend}=1kg$ is the pendulum mass, $L=2m$ indicates the length of the pendulum, ${m}_{cart}=5kg$ denotes the cart mass, ${K}_{d}=10Ns/m$ is the damping parameter, and $g=9.81m/s^{2}$ represents the gravity acceleration. The force $F$ is the control input, and $T=0.1s$ is the sampling time.

Now, the safety constraint and the input constraint are expressed for the model \eqref{Cart-inverted pendulum} as
\begin{equation*}
  \begin{aligned}
    \label{pendulum's safety}
    	-5\le z\le 5,
  \end{aligned}
\end{equation*}
\begin{equation*}
  \begin{aligned}
    \label{pendulum's constraints}
    	-100\le F\le 100.
  \end{aligned}
\end{equation*}
where using \eqref{Safe set} and \eqref{Mean Value}, the CBF and the RCBF are
\begin{equation*}
  \begin{aligned}
    \label{pendulum's CBF}
    	h(x)=\left\| x \right\|-0.5,
  \end{aligned}
\end{equation*}
\begin{equation*}
  \begin{aligned}
    \label{pendulum's RCBF}
    	{h}_{r}(\hat{x})=\left\| \hat{x} \right\|-0.5+\eta (\varepsilon_w + \varepsilon_s + \varepsilon_c),
  \end{aligned}
\end{equation*}
where $\eta=1$, $\varepsilon_w=0.01$, $\varepsilon_s=0.01$, and $\varepsilon_c=0.01$ are considered for the RCBF. It should be mentioned that the RCBF is only applied on $z$.

\begin{figure}[!h]
     \centering
     \includegraphics[width=0.59\linewidth]{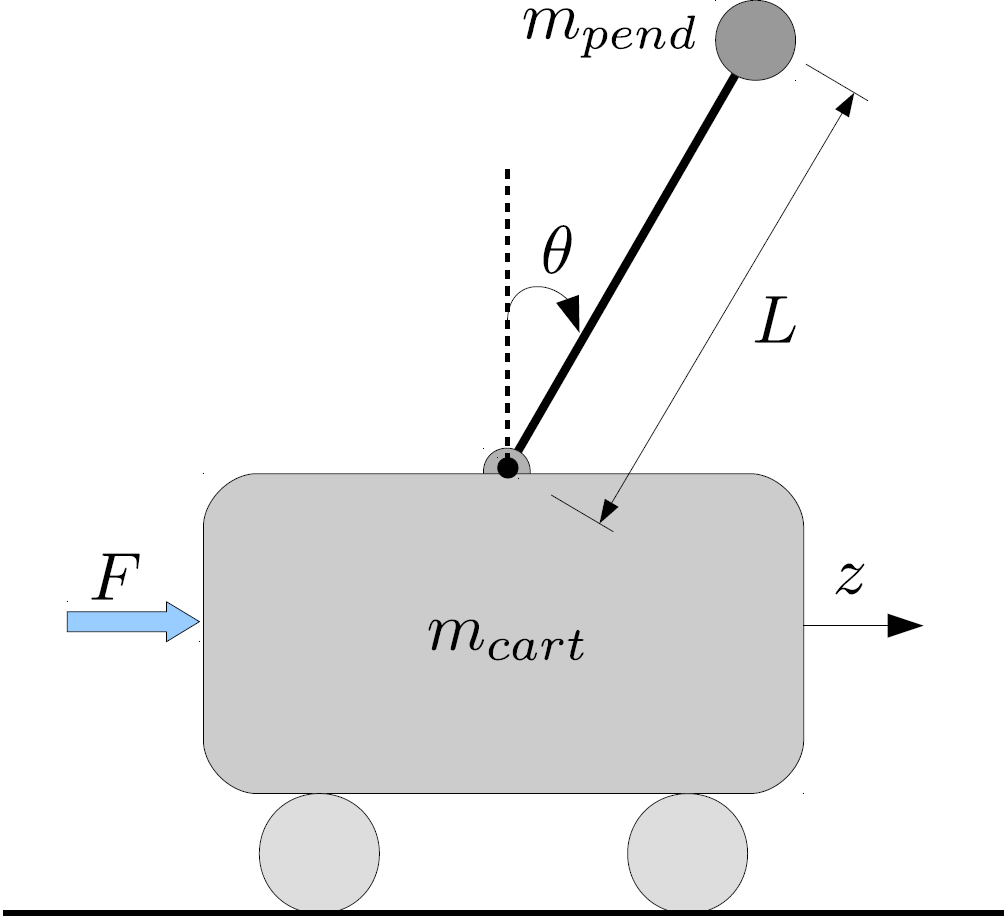}
    \caption{Cart-inverted pendulum.}
     \label{Cart-inverted pendulum fig}
 \end{figure}
 
The online data-driven safe predictive control is applied to the cart-inverted pendulum \eqref{Cart-inverted pendulum}, which yields 1) 97.54\% accuracy (validation performance) for the system identification part with 27000 training data, 3000 validation data, and 10 clusters, and 2) 98.66\% accuracy for the control policy learning part with 90000 training data, 10000 validation data, and 20 clusters. We collect the data set under random inputs/states for the system identification/policy learning. We have compared the performance of our proposed online data-driven STF controller with tube-based NMPC \cite{mayne2011tube} which is one of the state-of-the-art robust NMPC schemes to handle model uncertainties. The tube-based NMPC consists of a nominal controller that generates a central path and an ancillary controller that endeavours to steer the trajectories of the uncertain system to the central path. Using an NMPC that minimizes the cost of the deviation between the trajectories of the real system and the nominal system, the ancillary controller maintains the state of the real system in a tube whose centre is the trajectory of the nominal system. At each time step, the controller solves two optimal control problems, one which solves a standard problem for the nominal system with tightened distance constraints (the solution of which defines a central path) and an ancillary problem which keeps the actual trajectories close to the nominal trajectory.

For the nonlinear system \eqref{Cart-inverted pendulum} under an external disturbance ${w}(k)=-0.1+0.2\times rand(k)$, Fig. \ref{Cart-inverted pendulum Control input} indicates the signal $F$ for the RCBF-based NMPC, the tube-based NMPC, the offline safe data-driven predictive control, and the online safe data-driven predictive control with four future state predictions $N=4$. For this example, the RCBF-based NMPC \eqref{RCBF-based NMPC} and the tube-based NMPC work with the well-known model \eqref{Cart-inverted pendulum}, and we adapt the RCBF-based NMPC using the feedback controller \eqref{Online Adaptation} to minimize the performance loss caused by the considered external disturbance ${w}(k)$. Moreover, one can see that the proposed online adaptive control policy corrects the offline safe data-driven predictive control and approximates the policy \eqref{RCBF-based NMPC} better. From Figs. \ref{Cart-inverted pendulum Control input} and \ref{Cart-inverted pendulum Outputs}, one can see that Condition 2 is satisfied such that the considered nonlinear system converges to an equilibrium point $(x^{\prime}_e,u^{\prime}_e)$ using the offline safe data-driven predictive control; however, it is not same as the equilibrium point $(x_e,u_e)$ obtained by the policy \eqref{RCBF-based NMPC}. Therefore, the efficacy of the online adaptive control policy is clearly demonstrated so that it makes same equilibrium point $(x_e,u_e)$ for the real system. From Fig. \ref{Cart-inverted pendulum Outputs}, one can see that the tube-based NMPC achieves the system safety in the presence of the external disturbance ${w}(k)$; however, it reaches the constraint due to considering the distance constraint in the optimization problem. On the other hand, the RCBF-based NMPC and the online adaptive control policy keep the system far from the unsafe region and provide safer behavior for the real system. It is worth noting that although the offline safe data-driven predictive control shows performance loss due to the model uncertainties, it also satisfies the system safety.

\begin{figure}[!h]
     \centering
     \includegraphics[width=0.99\linewidth]{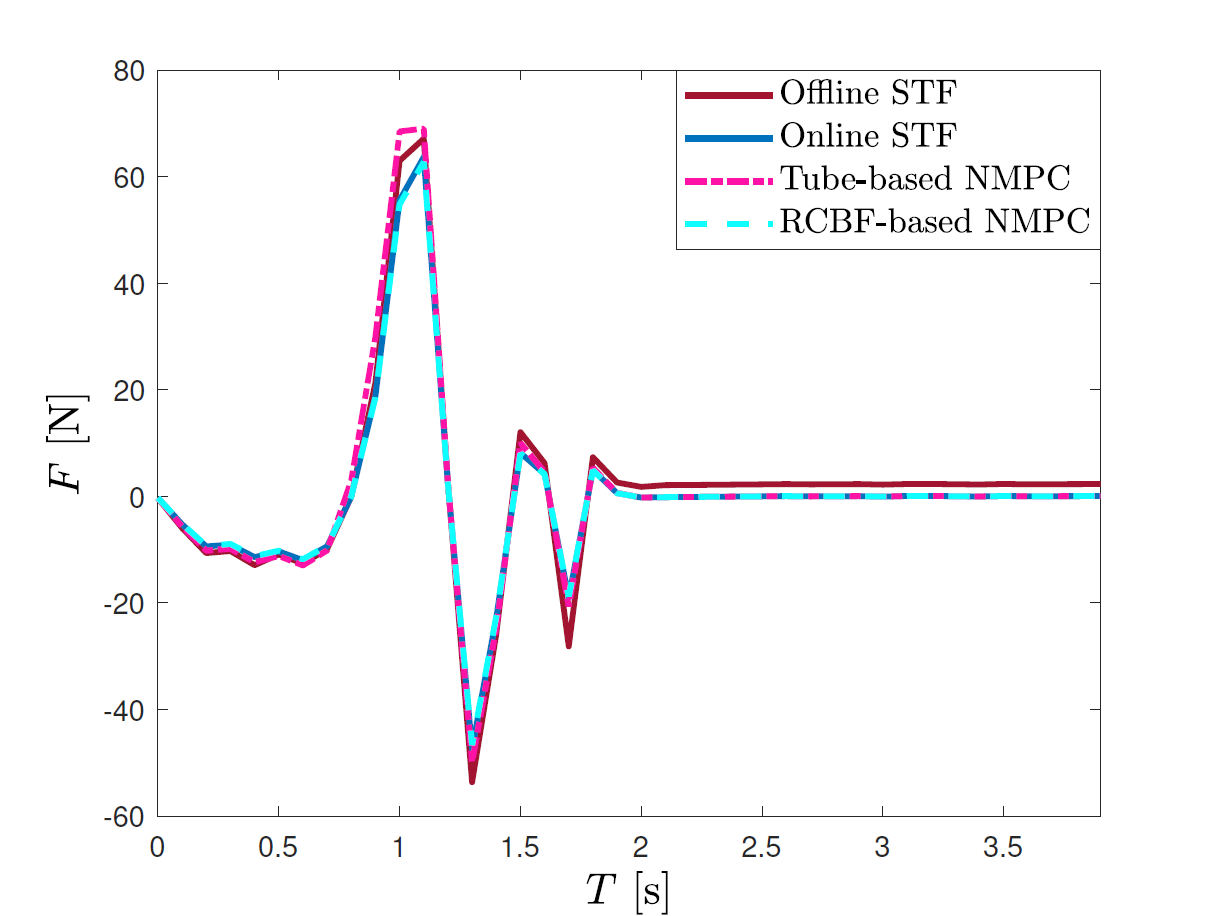}
    \caption{Control input of cart-inverted pendulum.}
     \label{Cart-inverted pendulum Control input}
 \end{figure}
 
 
 \begin{figure}[!h]
     \centering
     \includegraphics[width=0.99\linewidth]{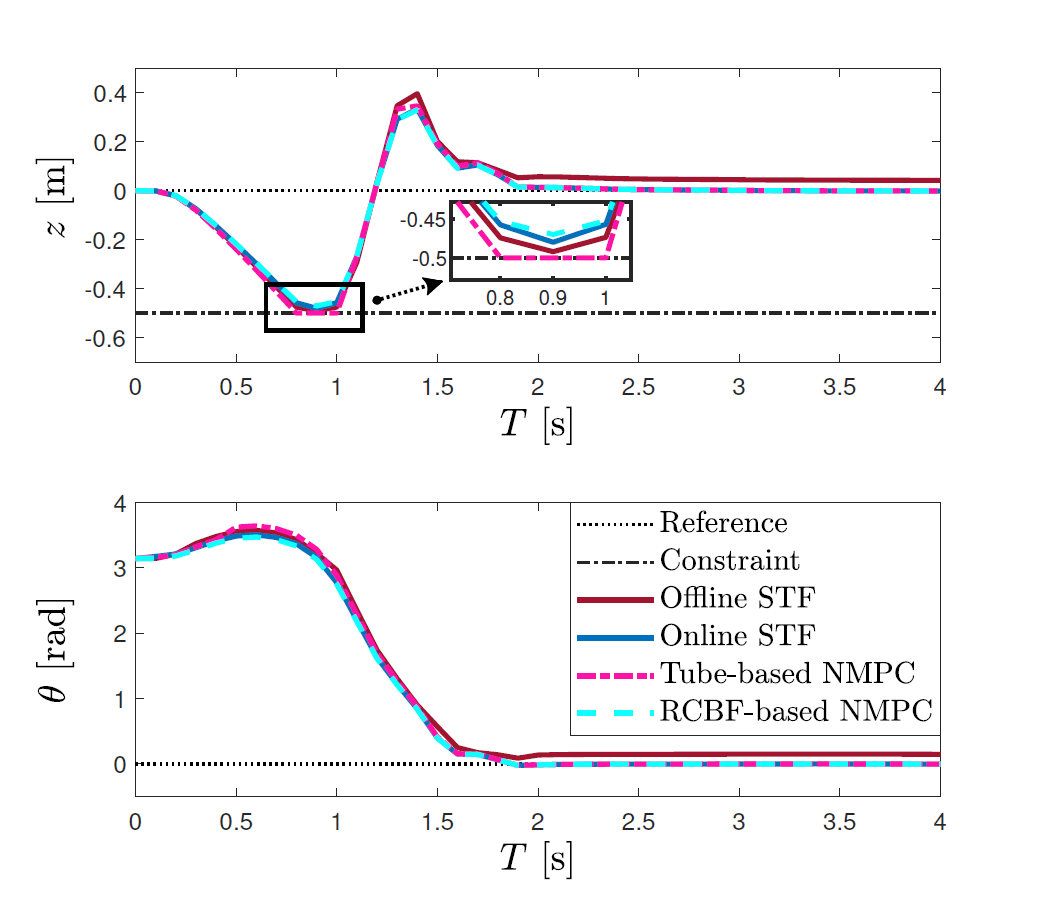}
    \caption{Outputs of cart-inverted pendulum.}
     \label{Cart-inverted pendulum Outputs}
 \end{figure}
 
  \begin{figure}[!h]
     \centering
     \includegraphics[width=0.95\linewidth]{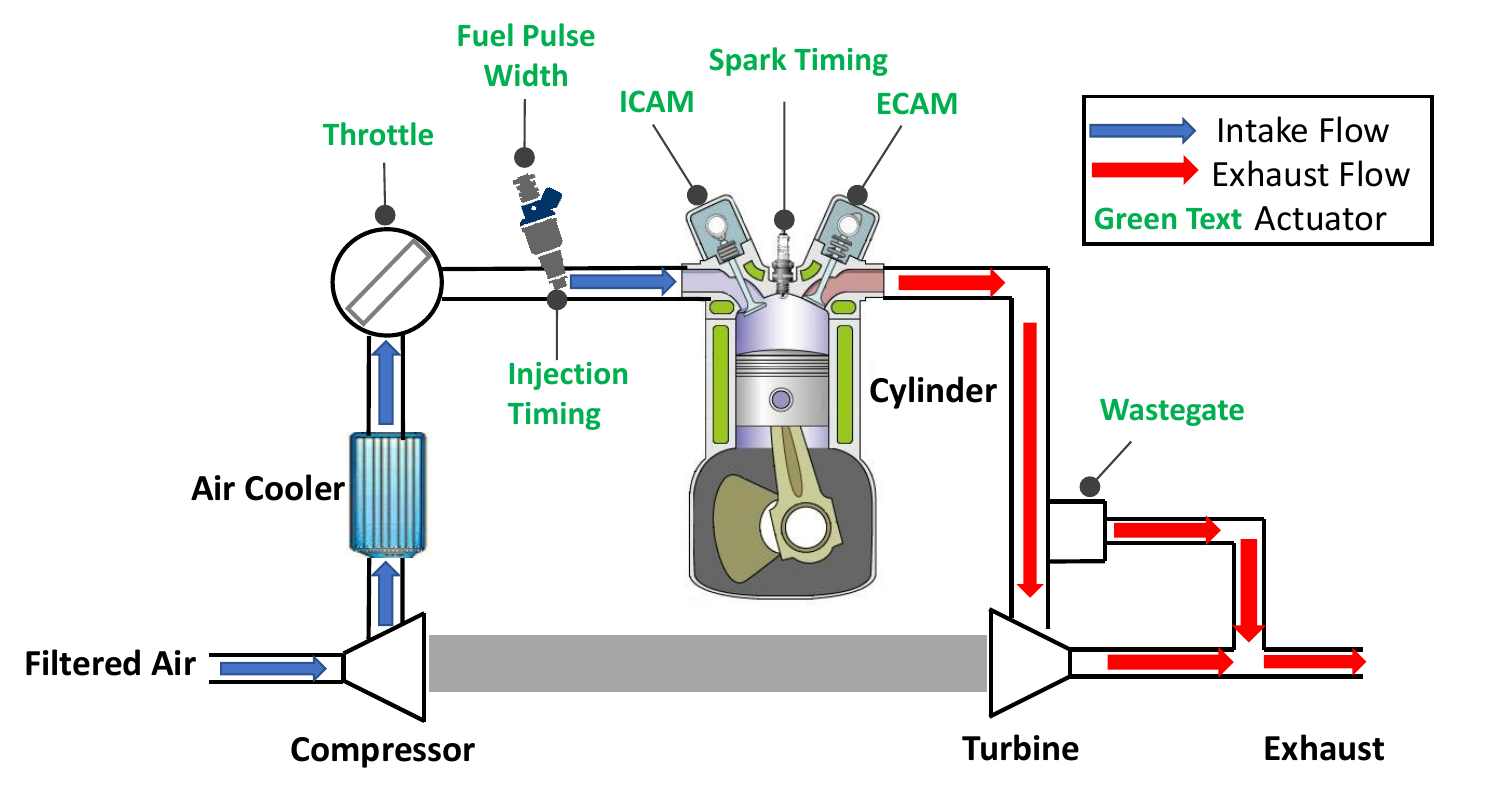}
    \caption{Turbocharged internal combustion engine.}
     \label{Turbocharged internal combustion engine}
 \end{figure}

\subsection{Turbocharged Internal Combustion Engine}
In this subsection, we employ the proposed control framework for a turbocharged internal combustion engine (see Fig. \ref{Turbocharged internal combustion engine}) \cite{chen2020online}. For this system, we have four control inputs as throttle position, intake cam position, exhaust cam position, and spark timing. 
Moreover, the system outputs are considered as the vehicle speed and the fuel consumption rate. Obviously, the considered engine system represents a complex nonlinear input–output relationship, which makes the system identification challenging. For this case, we aim at minimizing the fuel consumption rate and tracking a desired reference trajectory for the vehicle speed. The online data-driven safe predictive control is applied to the turbocharged internal combustion engine based on experimental data collected by Ford Motor Company on their engine platform, which yields 1) 96.12\% accuracy for the system identification part with 27000 training data, 3000 validation data, and 20 clusters, and 2) 98.05\% accuracy for the control policy learning part with 90000 training data, 10000 validation data, and 20 clusters. We collect the data set under multiple short input trajectories and random states for the system identification and policy learning, respectively.


For this case, the RCBF-based NMPC policy \eqref{RCBF-based NMPC} works with the trained model obtained by the STF-based concurrent learning, and we adapt it using the feedback controller in \eqref{Online Adaptation} to minimize the performance loss caused by the system identification error and the unknown disturbance for the real system. Considering four future state predictions $N=4$, Fig. \ref{Turbocharged internal combustion engine Control Inputs} shows the control input signals for the turbocharged internal combustion engine. Like the previous example, one can see that the proposed online adaptive control policy corrects the offline safe data-driven predictive control and learns the policy \eqref{RCBF-based NMPC} better. Fig. \ref{Turbocharged internal combustion engine Outputs} indicates the outputs of the turbocharged internal combustion engine, where the fuel consumption is minimized, and the vehicle speed tracks the desired reference while it satisfies the constraint. Clearly, the offline safe data-driven predictive control causes a performance loss for the real system; however, the online safe data-driven predictive control minimizes the performance loss by removing the KKT deviations caused by the control learning error and state perturbations caused by the system identification error and unknown disturbance. Fig.~\ref{Number of clusters and local models} shows the distribution of the identification performance along the number of clusters and local models for the turbocharged internal combustion engine. As it is obvious from Fig. \ref{Number of clusters and local models}, there is no major change on the performance after 20 clusters; thus, we have considered this number for the engine vehicle identification. Moreover, Fig. \ref{Delay} illustrates the distribution of the identification performance along the input delay $d_u$ and the output delay $d_y$ for the turbocharged internal combustion engine. One can see that $d_u=3,d_y=2$ makes the best identification performance for the engine vehicle, where these values are considered for the engine vehicle identification. In Figs. 10 and 11, validation performance demonstrates the accuracy of the identified model (trained by training data) to act as the real system, which is analyzed using validation data. Figs. 10 and 11 illustrate the validation performance when training the model with different numbers of clusters and input and output delays, respectively. In Fig. 11, to evaluate the validation performance for different input and output delays, we collect multiple short input-output trajectories as the training data set and consider the delay terms in the basis function of the trained model.

\begin{figure}[!h]
     \centering
     \includegraphics[width=0.99\linewidth]{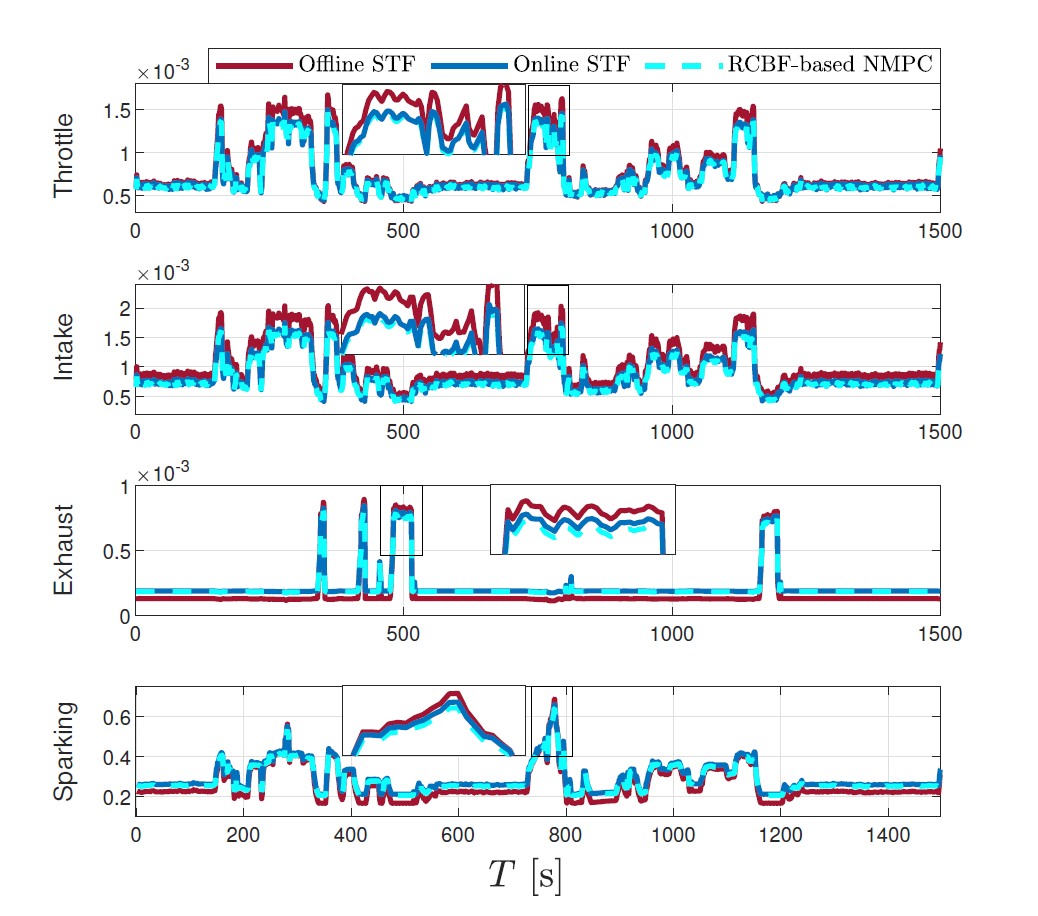}
    \caption{Control inputs of engine vehicle.}
     \label{Turbocharged internal combustion engine Control Inputs}
 \end{figure}

 \begin{figure}[!h]
     \centering
     \includegraphics[width=0.99\linewidth]{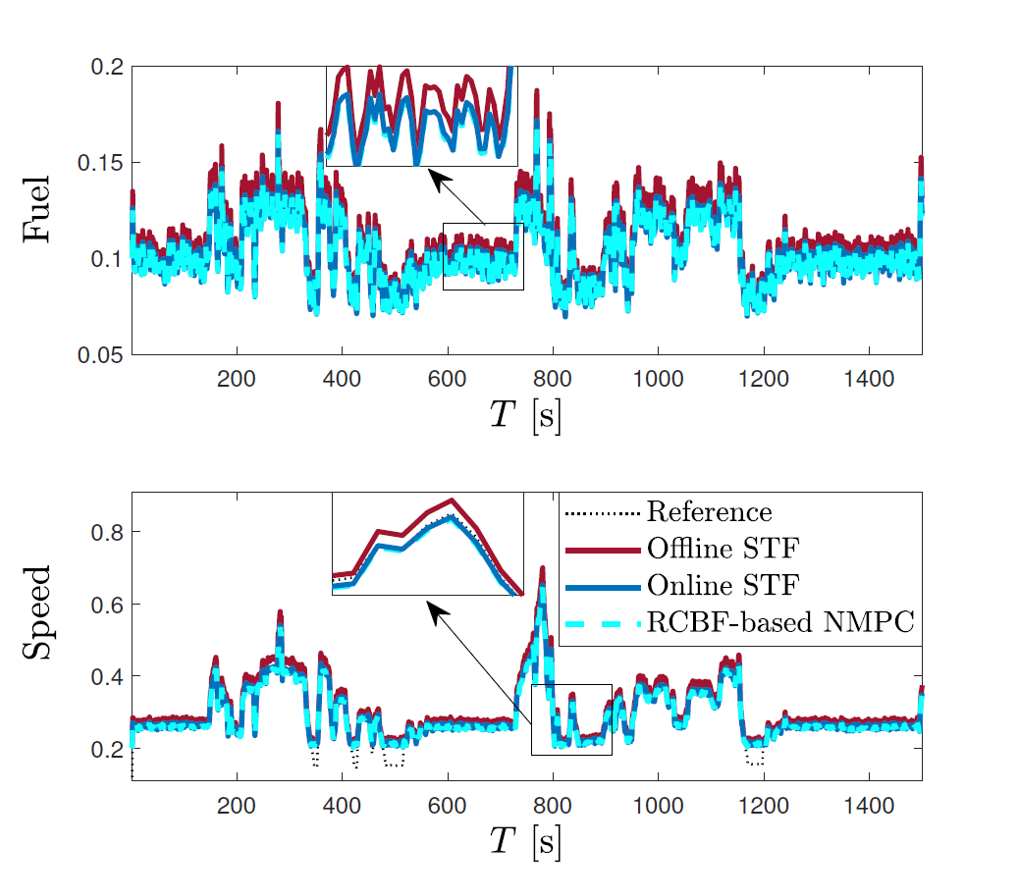}
    \caption{Outputs of engine vehicle.}
     \label{Turbocharged internal combustion engine Outputs}
 \end{figure}
 
  \begin{figure}[!h]
     \centering
     \includegraphics[width=0.99\linewidth]{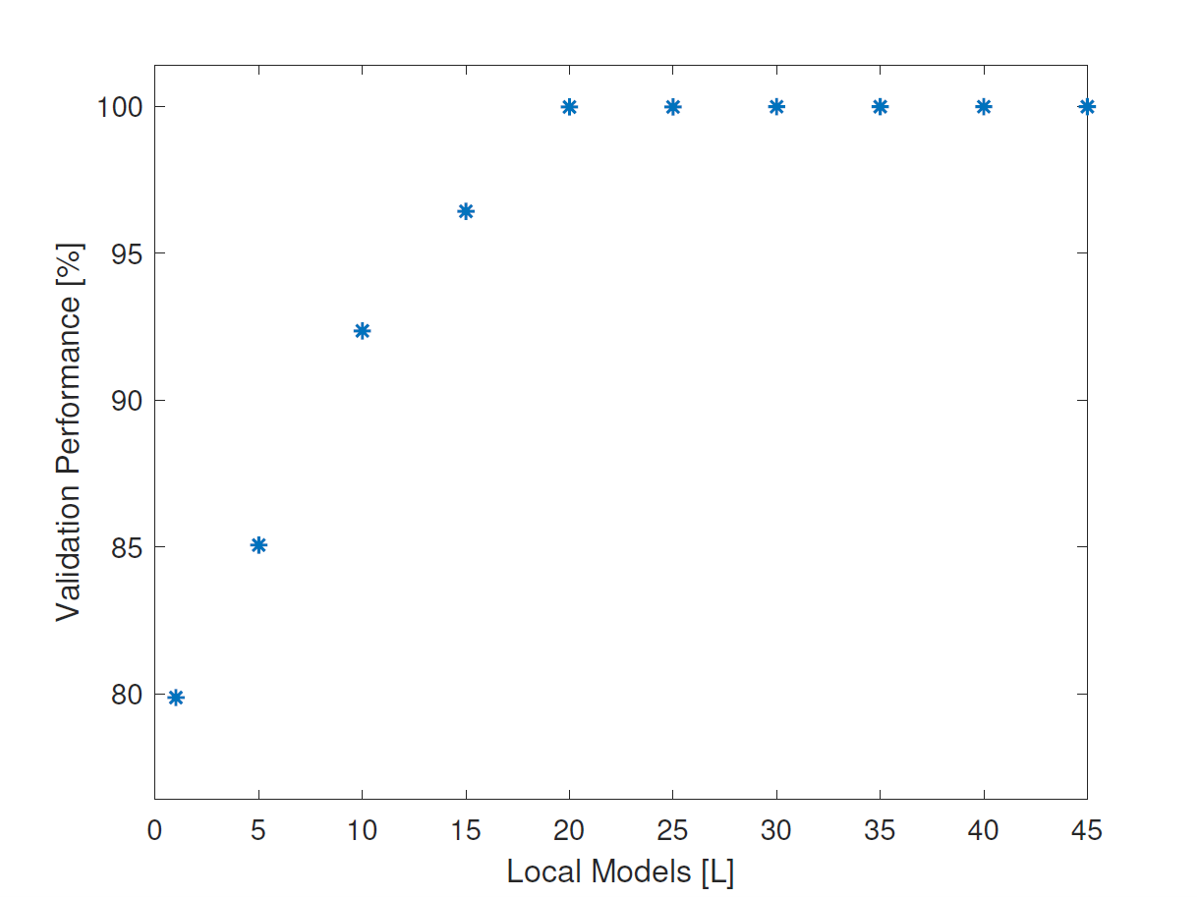}
    \caption{Number of clusters and local models for engine vehicle identification.}
     \label{Number of clusters and local models}
 \end{figure}
 
 \vspace{3 mm}
 
   \begin{figure}[!h]
     \centering
     \includegraphics[width=0.99\linewidth]{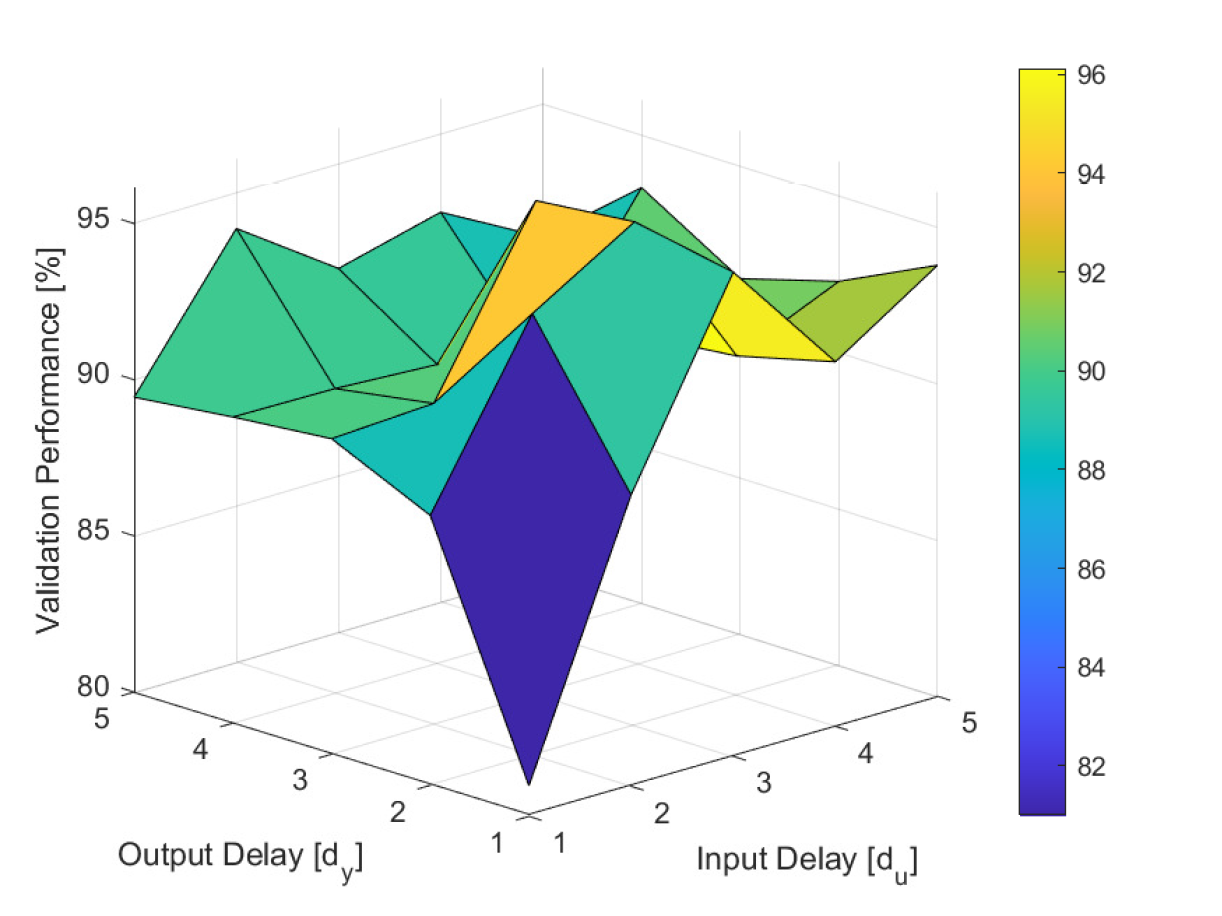}
    \caption{Input delay and output delay for engine vehicle identification.}
     \label{Delay}
 \end{figure} 

For the offline part, Table II presents the validation performance and computational cost of the STF-based concurrent learning in comparison with the NNs and the GPR for the engine vehicle identification, and Table III demonstrates the same task for the RCBF-based NMPC policy learning. We evaluate the validation performance of each function approximator using the best fitting rate (BFR), which analyzes the goodness of fit between the validation data (i.e., measured data) and the simulated output of the trained model (or trained policy) based on the normalized root mean squared error (NRMSE). Thus, the performance represents $(1-\frac{\lVert y - \hat{y} \rVert}{\lVert y - mean(y) \rVert}) \times 100 \hspace{1 mm} \%$ and $(1-\frac{\lVert u_{mpc} - \tilde{u} \rVert}{\lVert u_{mpc} - mean(u_{mpc}) \rVert}) \times 100 \hspace{1 mm} \%$ for Tables II and III, respectively, and 100\% performance means that the simulated output from the trained model (or trained policy) is perfectly matched with the measured data. It is worth noting that the measured data for Tables II and III are the collected system output from the engine vehicle and the collected control input from the RCBF-based NMPC, respectively. As shown in Tables II and III, the STF-based concurrent learning effectively reduces the computational cost of the learning process while it shows high performance compared to the NNs and the GPR. Moreover, for the online part, the online safe data-driven predictive control is compared with the policy \eqref{RCBF-based NMPC} for various future predictions $N$ in Table IV. This table demonstrates the tracking performance of each controller using the NRMSE between the desired reference trajectory and the system output. Thus, the tracking performance represents $(1-\frac{\lVert r - y \rVert}{\lVert r - mean(r) \rVert}) \times 100 \hspace{1 mm} \%$, and 100\% performance means that the system output is perfectly matched with the desired reference trajectory. Obviously, the online safe data-driven predictive control provides high performance to track the desired reference trajectory while it effectively reduces the computational cost of the NMPC.

\vspace{3 mm}
\begin{table}[!ht]
\centering
 \caption{Comparison of Performance and Computational Cost for System Identification}
\begin{tabular}{ |p{3.4cm}|p{1.7cm}|p{2.2cm}|  }
\hline\hline
Function Approximator & Performance & Time (per loop) \\
\hline
NNs & $96.15\%$ & $32.8 \hspace{1 mm} ms$ \\\hline
GPR & $96.48\%$ & $99.3 \hspace{1 mm} ms$ \\\hline
STF & $96.12\%$ & $04.3 \hspace{1 mm} ms$ \\\hline
\hline
\end{tabular}
\end{table}
 
\vspace{3 pt}
\begin{table}[!ht]
\centering
 \caption{Comparison of Performance and Computational Cost for NMPC Policy Learning}
\begin{tabular}{ |p{3.4cm}|p{1.7cm}|p{2.2cm}|  }
\hline
\hline
Function Approximator & Performance & Time (per loop) \\
\hline
NNs & $98.12\%$ & $19.4 \hspace{1 mm} ms$ \\\hline
GPR & $98.57\%$ & $58.1 \hspace{1 mm} ms$ \\\hline
STF & $98.05\%$ & $03.5 \hspace{1 mm} ms$ \\\hline
\hline
\end{tabular}
\end{table}

\begin{table}[!ht]
\centering
 \caption{Comparison of Performance and Computational Cost for Different Controllers}
\begin{tabular}{ |p{3.4cm}|p{1.7cm}|p{2.2cm}|  }
\hline\hline
Control Policy & Performance & Time (per loop) \\
\hline
NMPC (N=1) & $98.76\%$ & $05.7 \hspace{1 mm} ms$ \\\hline
NMPC (N=4) & $99.23\%$ & $25.2 \hspace{1 mm} ms$ \\\hline
NMPC (N=8) & $99.86\%$ & $68.1 \hspace{1 mm} ms$ \\\hline
Online STF (N=1) & $98.66\%$ & $66.2 \hspace{1 mm} \mu s$ \\\hline
Online STF (N=4) & $99.15\%$ & $78.9 \hspace{1 mm} \mu s$ \\\hline
Online STF (N=8) & $99.79\%$ & $90.8 \hspace{1 mm} \mu s$ \\\hline
\hline
\end{tabular}
\end{table}

\section{Conclusions}
\label{Sec5}
In general, safety refers to the system's ability to operate without causing hazardous conditions, physical damage, or harm to humans, infrastructure, or the environment (such as collisions in autonomous vehicles). A safe controller ensures that operational constraints (e.g., speed, temperature, pressure, and voltage) remain within acceptable limits to prevent failures or dangerous consequences. MPC is a widely used optimal control strategy that accounts for system safety; however, MPC can be vulnerable to model learning errors, disturbances, and cyber-attacks. In \cite{yang2023resilient}, a resilient MPC framework is proposed, which not only resists disturbances but also detects cyber-attacks (e.g., false data injection) and reconfigures MPC actions to maintain safety. However, in our paper, we discussed that MPC can drive the system too close to unsafe regions, which is an undesirable behavior. To address this, we enhance the safety properties of MPC by introducing the CBF-based MPC, which ensures the system remains farther from unsafe regions. Furthermore, by incorporating the robust CBF-based MPC, we strengthen our safe control framework to handle model learning errors and disturbances more effectively. It is important to note that in this paper, we do not consider cyber-attacks, as our primary focus is on developing a safe data-driven predictive control framework that maintains system safety in the presence of model learning errors and disturbances. A resilient extension would be required when security threats are a concern. To summarize, we proposed a safe data-driven predictive control framework, which includes i) a discrete-time STF-based concurrent learning for system identification and control policy learning, ii) a RCBF-based NMPC policy, and iii) an online adaptation law based on KKT sensitivity analysis and feedback control. The proposed control framework was employed for the cart-inverted pendulum as well as the automotive engine with promising results demonstrated. 
For the automotive engine, we have experimentally collected the I/O data from the turbocharged internal combustion engine and identified a nominal model for the system using the collected data. The proposed safe data-driven predictive control is applied on the obtained identified model and demonstrates a reasonable performance as shown in the simulation results. We acknowledge a few limitations of our method as follows. Like the NNs and the GPR, the STF framework requires comprehensive data collection to ensure an adequate coverage of operating conditions. Moreover, in the control design, a bounded external disturbance was assumed. Future work will include addressing the mentioned shortcomings by exploring a finite sample approach to reduce the required collected data for the STF, extending the control framework with more general unbounded stochastic disturbances, and carrying out a formal discussion on the recursive feasibility of the optimization problem (e.g. with an iteration approach).

\begin{appendices}
\section{(Transferring STF Model to State-Space Model)}
Considering \eqref{NARX Model}-\eqref{local models}, the input vector of the STF function approximator, i.e., ${U}_{stf}(k+1)$, can be written in the format of $[u(k);x(k)]$ as
\begin{equation*}
  \begin{aligned}
    \label{local models state}
  &{f}_{i}(k+1) = {A}_{i}{U}_{stf}(k+1)+{b}_{i}+\omega_i(k+1) \\ 
  & \hspace{13.5 mm} = A_{i_2} x(k)+{A}_{i_1}u(k)+{b}_{i}+\omega_i(k+1),\\
  \end{aligned}
\end{equation*}
where ${A}_{i_1}$ is the first element of ${A}_{i}$, and $A_{i_2}$ represents the rest of the elements, i.e.,
\begin{equation*}
\small
  \begin{aligned}
    \label{parameters state}
    &{A}_{i} = [a_{i_1}, a_{i_2}, a_{i_3}, \ldots, a_{i_{d_u+d_y}}],\\
  &A_{i_1} = a_{i_1},{A}_{i_2} = [a_{i_2}, a_{i_3}, \ldots, a_{i_{d_u+d_y}}],\\
  &x(k) = [u(k-1); \ldots; u(k-d_{u}+1);y(k); \ldots; y(k-d_{y}+1)].
  \end{aligned}
\end{equation*}
where $x(k)$ is considered as the states of the system, and $u(k)$ represents the control input.

Now, the nonlinear model \eqref{Composite Structure} is expressed as
\begin{equation*}
\small
  \begin{aligned}
    \label{Composite Structure state}
  y(k+1)=\sum\limits_{i=1}^{L} {\alpha }_{i}(k+1) (A_{i_2} x(k)+{A}_{i_1}u(k)+{b}_{i}+\omega_i(k+1)), 
  \end{aligned}
\end{equation*}
where ${\alpha }_{i}(k+1)$ represents ${\alpha }_{i}([u(k);x(k)],{\psi }_{i})$.

Using \eqref{parameters state} and \eqref{Composite Structure state}, the state-space model \eqref{system} is given as follows:
\begin{equation*}
  \begin{aligned}
    \label{state-space}
    & x(k+1) = f(x(k),u(k))+w(k),\\
    & \hspace{13 mm} = A_{t_2}(k) x(k)+{A}_{t_1}(k) u(k)+{b}_{t}(k)+w(k),
  \end{aligned}
\end{equation*}
where $A_{t_2}$, ${A}_{t_1}$, and ${b}_{t}$ are nonlinear matrices as
\begin{equation*}
  \begin{aligned}
    \label{state-space matrices}
    & A_{t_2}(k)=\begin{bmatrix} 
    0_{n_u \times n_u(d_u -1)} & 0_{n_u \times n_y(d_y)}\\ 
    \tau_1 & 0_{n_u(d_u -2) \times n_y(d_y)}\\
    \rho_1(k) & \rho_2(k)\\
    0_{n_y(d_y -1) \times n_u(d_u-1)} & \tau_2
    \end{bmatrix},\\
    & A_{t_1}(k)=\begin{bmatrix} 
    I_{n_u}\\ 
    0_{n_u(d_u -2) \times n_u}\\
    \sum\limits_{i=1}^{L} {\alpha }_{i}([u(k);x(k)],{\psi }_{i}) {A}_{i_1},\\
    0_{n_y(d_y -1) \times n_u}
    \end{bmatrix},\\
    & b_{t}(k)= \begin{bmatrix} 
    0_{n_u \times 1}\\ 
    0_{n_u(d_u -2) \times 1}\\
    \sum\limits_{i=1}^{L} {\alpha }_{i}([u(k);x(k)],{\psi }_{i}) b_i,\\
    0_{n_y(d_y -1) \times 1}
    \end{bmatrix},\\
    & w(k)= \begin{bmatrix} 
    0_{n_u \times 1}\\ 
    0_{n_u(d_u -2) \times 1}\\
    \sum\limits_{i=1}^{L} {\alpha }_{i}([u(k);x(k)],{\psi }_{i}) \omega_i(k+1),\\
    0_{n_y(d_y -1) \times 1}
    \end{bmatrix},\\
  \end{aligned}
\end{equation*}
where
\begin{equation*}
  \begin{aligned}
    \label{state-space terms}
    & \tau_1 = [I_{n_u(d_u -2)} \hspace{5 mm} 0_{n_u(d_u -2) \times n_u}],\\
    & \tau_2 = [I_{n_y(d_y -1)} \hspace{5 mm} 0_{n_y(d_y -1) \times n_y}],\\
    & \rho_1(k) = \sum\limits_{i=1}^{L} {\alpha }_{i}([u(k);x(k)],{\psi }_{i}) [a_{i_2}, \ldots, a_{i_{d_u}}],\\
    & \rho_2(k) = \sum\limits_{i=1}^{L} {\alpha }_{i}([u(k);x(k)],{\psi }_{i}) [a_{i_{d_u+1}}, \ldots, a_{i_{d_u+d_y}}].
  \end{aligned}
\end{equation*}
\end{appendices}

\bibliographystyle{ieeetr}
\bibliography{References.bib}

\begin{IEEEbiography}[{\includegraphics[width=1.05in,height=1.2in]{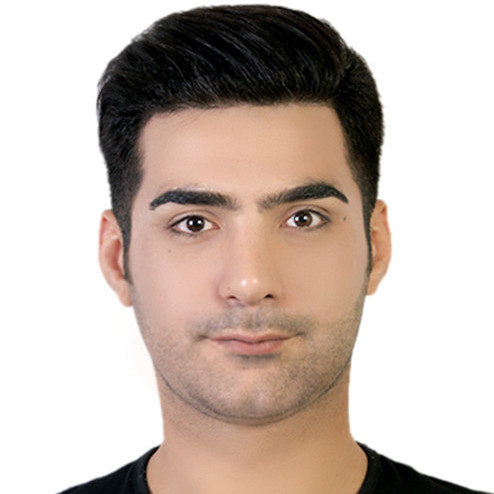}}]{Amin Vahidi-Moghaddam}
received the B.S. Degree in Mechanical Engineering from Shiraz University, Iran, in 2015, the M.S Degree in Mechanical Engineering from the University of Tehran, Iran, in 2017, and the Ph.D. degree in Mechanical Engineering from Michigan State University, USA, in 2024. He was a Researcher in the Department of Mechanical Engineering at the University of Tehran from 2017 to 2019. He is currently a Research Scientist in the Intelligent Thermal Control Team at Vitro. His current research interests include Model Predictive Control, Reinforcement Learning, Deep Learning, Optimization, Thermal Control, and Robot/Vehicle Motion Control.
\end{IEEEbiography}

\vspace{-9 mm}

\begin{IEEEbiography}[{\includegraphics[width=1.03in,height=1.25in]{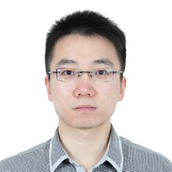}}]{Kaian Chen}
received his B.E. degree from Nantong University, China, in 2013, his M.S. degree in Electrical and Computer Engineering from The Ohio State University in 2017, and his Ph.D. degree in Mechanical Engineering from Michigan State University in 2022. He is currently a Battery Control Engineer at the Battery Energy Control Module (BECM) Team of Ford Motor Company. His work and research interests include nonlinear system identification and modeling, model predictive control, and their applications in battery energy management, health diagnostics, and charging status prediction.
\end{IEEEbiography}

\vspace{-9 mm}

\begin{IEEEbiography}[{\includegraphics[width=1.03in,height=1.25in]{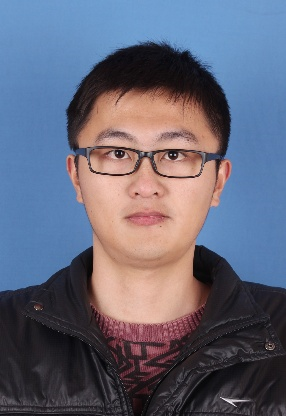}}]{Kaixiang Zhang}
received the B.E. degree in automation from Xiamen University, Xiamen, China, in 2014, and the Ph.D. degree in control science and engineering from Zhejiang University, Hangzhou, China, in 2019. From 2017 to 2018, he was a visiting doctoral student at the National Institute for Research in Computer Science and Automation, Rennes, France. He was a research associate with the Department of Mechanical Engineering at Michigan State University, where he is currently a fixed-term Assistant Professor. His research interests include visual servoing, robotics, and control theory and application.

\end{IEEEbiography}

\vspace{-9 mm}

\begin{IEEEbiography}[{\includegraphics[width=1.03in,height=1.25in]{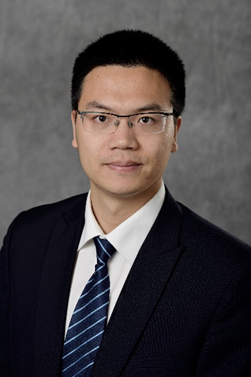}}]{Zhaojian Li}
received his B. Eng. degree from Nanjing University of Aeronautics and Astronautics in 2010. He obtained M.S. (2013) and Ph.D. (2015) in Aerospace Engineering (flight dynamics and control) at the University of Michigan, Ann Arbor. He is currently a Red Cedar Distinguished Associate Professor in the Department of Mechanical Engineering at Michigan State University. His research interests include Learning-based Control, Nonlinear and Complex Systems, and Robotics and Automated Vehicles. He is a senior member of IEEE and a recipient of the NSF CAREER Award.

\end{IEEEbiography}

\vspace{-9 mm}

\begin{IEEEbiography}[{\includegraphics[width=1.03in,height=1.25in]{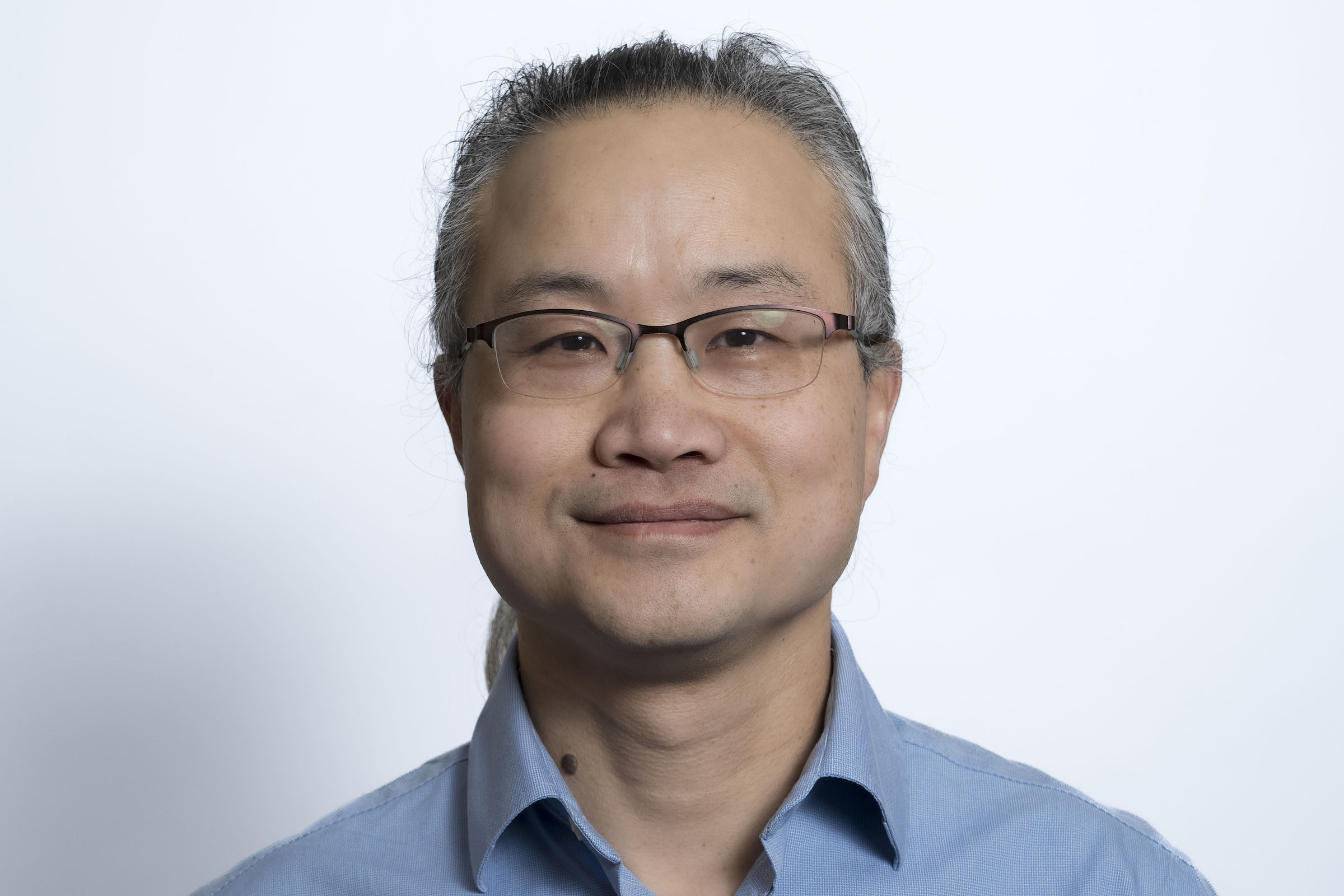}}]{Yan Wang}
received his B.S. and M.S. degrees from Tsinghua University in China, and the Ph.D. degree in Mechanical Engineering from the University of California, Santa Barbara, CA, USA. He spent 20+ years working on different automotive control problems, spanning from actuator design and mechatronics controls, to engine optimization, controls, and calibrations, and later on in data-driven control in automotive applications. His main research interests include the application and real world implementation of advanced/modern control methods, including robust control, adaptive control, system identification, optimal control, model predictive control, and connectivity enabled data analytics and machine learning, on vehicle design, control, calibration, and diagnostics/prognostics.
\end{IEEEbiography}

\vspace{-9 mm}

\begin{IEEEbiography}[{\includegraphics[width=1.03in,height=1.25in]{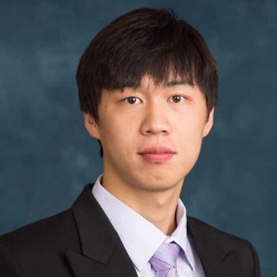}}]{Kai Wu}
received his B. Eng. degree from Tianjin University of Automation in 2011. He obtained a M.S. degree in Electrical Engineering and Computer Science (EECS) in 2013, and a Ph.D. in Mechanical Engineering in 2018, both from the University of Michigan, Ann Arbor. He is currently a Data Science Manager in Global Data Insights \& Analytics (GDI\&A) at Ford Motor Company. His work and research interests include charging and energy services, data-driven decision-making systems, and e-mobility. He is a member of IEEE and SAE and was awarded the Henry Ford Technology Award in 2024, which is the highest technology award at Ford Motor Company.
 
\end{IEEEbiography}

\vfill
 
\end{document}